\documentclass[a4paper,twocolumn,superscriptaddress,11pt,accepted=2018-11-28]{quantumarticle}

\usepackage[T1]{fontenc}
\usepackage[latin9]{inputenc}
\usepackage{color}
\usepackage[english]{babel}
\usepackage{amsthm}
\usepackage{amsmath}
\usepackage{amssymb,mathtools}
\usepackage{graphicx}
\usepackage{array,tabulary} 
\usepackage{hyperref}
\hypersetup{colorlinks,linkcolor=myurlcolor,citecolor=myurlcolor,urlcolor=myurlcolor}
\usepackage{txfonts}
\usepackage{colortbl}\definecolor{myurlcolor}{rgb}{0,0,0.7}
\usepackage[numbers,sort&compress]{natbib}

\makeatletter
\@ifundefined{textcolor}{}
{%
 \definecolor{BLACK}{gray}{0}
 \definecolor{WHITE}{gray}{1}
 \definecolor{RED}{rgb}{1,0,0}
 \definecolor{GREEN}{rgb}{0,1,0}
 \definecolor{BLUE}{rgb}{0,0,1}
 \definecolor{CYAN}{cmyk}{1,0,0,0}
 \definecolor{MAGENTA}{cmyk}{0,1,0,0}
 \definecolor{YELLOW}{cmyk}{0,0,1,0}
 }

\newcolumntype{K}[1]{>{\centering\arraybackslash}m{#1}}

\theoremstyle{plain}
\newtheorem{thm}{\protect\theoremname}
\theoremstyle{plain}

\ifx\proof\undefined
\newenvironment{proof}[1][\protect\proofname]{\par
\normalfont\topsep6\p@\@plus6\p@\relax
\trivlist
\itemindent\parindent
\item[\hskip\labelsep
\scshape
#1]\ignorespaces
}{%
\endtrivlist\@endpefalse
}
\providecommand{\proofname}{Proof}
\fi
\theoremstyle{plain}
\newtheorem{lem}[thm]{\protect\lemmaname}

\newtheorem{theorem}[thm]{\protect\theoremname}
\newtheorem{defi}[thm]{\protect\definitionname}

\providecommand{\lemmaname}{Lemma}
\providecommand{\theoremname}{Theorem}
\providecommand{\definitionname}{Definition}
\providecommand{\propositionname}{Proposition}

\renewcommand\qedsymbol{$\blacksquare$}
\newenvironment{proof-of}[1][{\hspace{-\blank}}]{{\medskip\noindent\textit{Proof~{#1}.\ }}}{\hfill\qedsymbol}

\newcommand{\tr}{{\operatorname{Tr\,}}}
\newcommand{\Tr}{{\operatorname{Tr\,}}}
\newcommand{\id}{{\mathbb{I}}}
\newcommand{\e}{{\operatorname{e}}}
\newcommand{\dd}{{\operatorname{d}}}

\newcommand{\nc}{\newcommand}
\nc{\rnc}{\renewcommand}

\nc{\cT}{{\cal T}}
\nc{\1}{{\openone}}
\nc{\ox}{\otimes}
\nc{\cH}{{\cal H}}
\nc{\cM}{{\cal M}}
\DeclareMathOperator*{\argmin}{arg\,min}

\def\ket#1{| #1 \rangle}
\def\bra#1{\langle  #1 |}

\def\proj#1{| #1 \rangle\!\langle #1 |}
\nc{\ketbra}[2]{|#1\rangle\!\langle#2|}
\nc{\avg}[1]{\langle#1\rangle}

\newcommand{\norm}[1]{\left\| #1 \right\|}

\newcommand{\ar}[1]{{#1}}
\newcommand{\aw}[1]{{#1}}

\newcommand{\EP}{entropy preserving }
\newcommand{\CP}{completely passive }

\begin{document}

\title{{Thermodynamics as a Consequence of \phantom{===========} Information Conservation}}

\author{Manabendra Nath Bera}
\orcid{0000-0002-8329-2656}
\email{mnbera@gmail.com}
\affiliation{ICFO -- Institut de Ciencies Fotoniques, The Barcelona Institute of Science and Technology, ES-08860 Castelldefels, Spain}
\affiliation{Max-Planck-Institut f\"ur Quantenoptik, D-85748 Garching, Germany}

\author{Arnau Riera}
\orcid{0000-0002-3271-7802}
\affiliation{ICFO -- Institut de Ciencies Fotoniques, The Barcelona Institute of Science and Technology, ES-08860 Castelldefels, Spain}
\affiliation{Max-Planck-Institut f\"ur Quantenoptik, D-85748 Garching, Germany}

\author{Maciej Lewenstein} 
\orcid{0000-0002-0210-7800}
\affiliation{ICFO -- Institut de Ciencies Fotoniques, The Barcelona Institute of Science and Technology, ES-08860 Castelldefels, Spain}
\affiliation{ICREA, Pg.~Lluis Companys 23, ES-08010 Barcelona, Spain} 

\author{Zahra Baghali Khanian}
\orcid{0000-0002-0892-7519}
\affiliation{ICFO -- Institut de Ciencies Fotoniques, The Barcelona Institute of Science and Technology, ES-08860 Castelldefels, Spain}
\affiliation{Departament de F\'isica: Grup d'Informaci\'o Qu\`antica, Universitat Aut\`onoma de Barcelona, ES-08193 Bellaterra (Barcelona), Spain}

\author{Andreas Winter}
\orcid{0000-0001-6344-4870}
\email{andreas.winter@uab.cat}
\affiliation{ICREA, Pg.~Lluis Companys 23, ES-08010 Barcelona, Spain}
\affiliation{Departament de F\'isica: Grup d'Informaci\'o Qu\`antica, Universitat Aut\`onoma de Barcelona, ES-08193 Bellaterra (Barcelona), Spain}

\begin{abstract}
Thermodynamics and information have intricate \aw{interrelations}. Often thermodynamics is considered to be the logical premise to justify that {\it information is physical} -- through Landauer's principle --, thereby also \aw{linking} information and thermodynamics. This approach towards information has been instrumental to understand thermodynamics of logical and physical processes, both in \aw{the classical and quantum domain}. In \aw{the present} work, we formulate thermodynamics as an exclusive consequence of information conservation. The framework can be applied to \aw{the} most general situations, beyond the traditional assumptions in thermodynamics: \aw{we allow} systems and \aw{thermal baths} to be quantum, of arbitrary sizes and even possessing inter-system correlations. 

Here, systems and baths are not treated differently, rather both are considered on \aw{an} equal footing. This leads us to introduce a ``temperature''-independent formulation of thermodynamics. We rely on the fact that, for a \aw{fixed} amount of information, measured by the von Neumann entropy, any system can be transformed to a state \aw{with the same entropy} that possesses minimal energy. This state, known as a {\it completely passive} state, acquires Boltzmann--Gibbs canonical form with an {\it intrinsic temperature}. We introduce the notions of bound and free energy and use them to quantify heat and work, respectively. 
Guided by \aw{the} principle of information conservation, we develop universal notions of equilibrium, heat and work, Landauer's principle and universal fundamental laws of thermodynamics. We demonstrate that the maximum efficiency of a quantum engine with a finite bath is in general lower than that of an ideal \aw{Carnot} engine. We introduce a resource theoretic framework for our {\it intrinsic temperature} based thermodynamics, within which we address the problem of work extraction and state transformations. 
Finally, the framework is extended to multiple conserved quantities.  
\end{abstract}


\maketitle

\section{Introduction}
Thermodynamics constitutes one of the basic foundations of modern science. It not only plays an important role in modern technologies, but offers \aw{a} basic understanding of \aw{a} vast range of natural phenomena. Initially, thermodynamics was developed \aw{phenomenologically}, to address the question \aw{of} how, and to what extent, heat can be converted into work. 
\aw{Later, with the developments of statistical mechanics, quantum mechanics and relativity, thermodynamics along with its fundamental laws was put on a formal and mathematically rigorous footing \cite{Gemmer09}. Beyond its original domain of conception, it has been applied in a wide range of contexts. It has been used to describe relativistic phenomena in astrophysics and cosmology, quantum effects in microscopic systems, or very complex systems in biology and chemistry, and even in music \cite{FS:heat-and-work}.}

The inter-relation between information and thermodynamics \cite{Parrondo15} is intricate, and has been studied in the context of Maxwell's demon \cite{Maxwell08, Leff90, Leff02, Maruyama09}, Szilard's engine \cite{Szilard29}, and Landauer's principle \cite{Landauer61, Bennett82, Plenio01, Rio11, Reeb14}. 
In recent years, classical and quantum information-theoretic approaches helped to understand thermodynamics in the domain of small classical and quantum systems \cite{Shannon48, Nielsen00, Cover05}.  
That stimulated a whole new perspective to tackle and extend thermodynamics beyond the standard classical domain. 
In fact, information theory has recently played an important role in understanding thermodynamics in the presence of inter-system and system-bath correlations \cite{Alicki13, Marti15, Bera16}, equilibration processes \cite{Short11, Goold16, Rio16, Gogolin16}, \aw{and the} foundations of statistical mechanics \cite{Popescu06}.
One of the most paradigmatic examples of this success is the formulation of quantum thermodynamics within the \aw{so-called} resource theoretic framework \cite{Brandao13}, which allows to reproduce standard thermodynamics in the asymptotic limit, when one processes infinitely many copies of the system under consideration. In the \aw{finite-copy, or even one-shot limit}, the resource theory reveals that the laws of thermodynamics have to be modified to \aw{capture finite-size, as well as quantum effects in thermodynamics} 
\cite{Dahlsten11,Aberg13,Horodecki13,Skrzypczyk14,Brandao15,Cwiklinski15,Lostaglio15,Egloff15,Lostaglio15a}. 

In \aw{the present} work, elaborating in the \aw{interrelations} between information and thermodynamics, we \aw{perform} an axiomatic construction of thermodynamics and identify the ``information conservation'' as the crucial underlying property of any theory that respects it. 
The \aw{cornerstone} of our construction is the notion of \emph{bound energy}, which we introduce as the amount of energy locked in a system that cannot be accessed (extracted) \aw{by using the} set of allowed operations. 
The bound energy obviously depends on the set of allowed operations: 
the larger this set is, the smaller the bound energy. 
We prove that, by taking ({\it i}) global \emph{entropy preserving operations} as the set of allowed operations and ({\it ii}) infinitely large thermal baths initially uncorrelated \aw{with} the system, our formalism reproduces standard thermodynamics \aw{in its resource theoretic form}.

In information theory, the \aw{von} Neumann entropy is the quantity that measures the amount of information in a system.
In this sense, entropy preserving operations are transformations that keep this information constant.
All fundamental physical theories, such as classical and quantum mechanics, share the property of conserving information \aw{at the basic level}. \aw{Namely, they have dynamics that is} deterministic and bijectively maps the set of possible configurations between any two instants of time. Thus, non-determinism can only appear when some degrees of freedom are ignored, leading to apparent information loss. In classical physics, this information loss is due to deterministic chaos and mixing in nonlinear dynamics. In quantum mechanics, this loss is intrinsic, and occurs due to measurement processes and non-local correlations \cite{BeraPhilo16}. \aw{Note} that the set of entropy preserving operations is larger than the dynamically reversible operations (unitaries) in the sense that
they conserve entropy but, unlike unitaries, not the individual probabilities. In other words, we demand only \emph{coarse-grained}, not microscopic information conservation. 
In the limit of many copies both coarse-grained and fine-grained information conservation become equivalent (see Sec.~\ref{sec:EPoperations}). While it is known that linear operations that are entropy preserving for all states are unitary \cite{Hulpke2006}, it is an open question to what extent coarse-grained
information conserving operations can be implemented in the single-copy \aw{setting}. 

The resource theory of thermodynamics can then be seen as one where condition ({\it i}) is sharpened, i.e. an extension of thermodynamics from coarse-grained to fine-grained information conservation, where the operations are global unitaries but still constrained to infinitely large thermal baths with a well defined temperature. 
Fluctuation theorems can also be \aw{considered} from this perspective.
There, the second law is obtained as a consequence of reversible transformations on initially thermal states or states with a well defined temperature \cite{Jarzynski2000, Esposito09, Esposito10, Strasberg17}.
In contrast, the aim of our work is instead to relax condition ({\it ii}), that is,
to generalize thermodynamics to be
valid for arbitrary environments, irrespective of being thermal, or much larger than the system.
This idea is illustrated in the table below.

\begin{center}
{\small
\begin{tabular}{K{1.6cm}K{2.7cm}K{2.7cm}}
& {\bf \phantom{Unitaries} Unitaries \phantom{Unitaries} (fine-grained information conservation)}  & {\bf Entropy preserving operations (coarse-grained information conservation)} \\
\hline
{\bf Large thermal bath}  & Resource theory of thermodynamics & Standard thermodynamics \\
\hline
{\bf Arbitrary environment} & ? & \aw{\emph{Present paper}} \\
\hline
\end{tabular}}
\end{center}
\vspace{-.25cm}\hspace{.3cm}
\vspace{.3cm}

The main obstruction against this generalization \aw{lies in} the fact that, when allowing for arbitrary states as environment, large amounts of resources can be pumped into the system leading to trivial ``resource'' theories. 
We are able to circumvent this problem thanks to the notion of bound energy, which intrinsically distinguishes accessible and 
\aw{non-accessible} energy.
Our formalism results in a theory in which systems and environments are treated \aw{on an} equal footing,
or in other words, in a ``temperature''-independent formulation of thermodynamics. 

Our work is complementary to other general approaches, where thermodynamics \aw{is} obtained after inserting some form of thermal state(s)
in general mathematical expressions, e.g. \cite{Sagawa2012}. 
While in these works the mathematical expressions for arbitrary states have no \emph{a priori} thermodynamic meaning, our construction is \aw{built} on a physically motivated quantity, the bound energy. 

\aw{The present temperature}-independent thermodynamics is essential in contexts in which the state of the bath can be affected by the system after exchange of heat (see \cite{Vulpiani2017} for a review on the notion of temperature). This can be due to the fact of \aw{either} having a relatively small environment compared to the system, or an environment simply not being thermal. In fact, in current experiments, environments do not have to be necessarily thermal, but can possess quantum coherence or correlations.

The entropy preserving operations make all states
with equal energies and entropies thermodynamically equivalent.
This allows for representing all the states and thermodynamic processes
in a simple energy-entropy diagram.
We exploit this geometric approach and give a diagrammatic representation for heat, work and other thermodynamic quantities. 
In this way we are able to reproduce several results of the literature, e.g. the resource theory of thermodynamics applicable for arbitrary quantum systems and environments \cite{Sparaciari16}.
Our formalism \aw{extends naturally} to scenarios with multiple conserved quantities.

\medskip
\aw{The structure of the rest of the paper is as follows: In the following
Section \ref{sec:EPoperations}, we discuss the class of \EP operations, on which
our theory is built. Then, in Section \ref{sec:bound+free}, we introduce
the fundamental notions of bound and free energy, and in Section \ref{sec:diagram}
define the energy-entropy diagram of a system. After that, we develop thermodynamics,
from the zeroth law (Section \ref{sec:zeroth}) and a discussion of the max-entropy
vs. min-energy principle (Section \ref{sec:athermality}), to the first law (Section \ref{sec:first})
and the second law of thermodynamics (Section \ref{sec:second}).
In Section \ref{sec:third-law} we discuss if, and why not, a third law can be 
formulated in our setting. 
Then, in Section \ref{sec:resource}, we move to the development of a resource theory
based on the previous formalism, which allows us to consider transformation rates 
between states and visualize the results in the energy-entropy diagram.
In Section \ref{sec:charges}, we outline how the theory so far developed
generalizes to the case of multiple conserved quantities.
Finally, we summarize our findings in Section \ref{sec:discuss}, and discuss the main obstructions
towards an extension of thermodynamics that would be valid both in
the single-shot scenario and for non-thermal environments.}

\section{Entropy preserving operations, entropic equivalence class and intrinsic temperature}
\label{sec:EPoperations}
The set of operations that we consider in this framework is the set of, so called, \emph{entropy preserving operations}.
Given a system initially in a state $\rho$, the set of entropy preserving operations are all the operations
that arbitrarily change the state, but keep its entropy constant
\begin{equation}
\rho \to \sigma \text{ s.t. } S(\rho)=S(\sigma) ,
\end{equation}
where $S(\rho)\coloneqq - \tr \rho \log \rho$ is the von Neumann entropy.
Importantly, an operation that acts on $\rho$ and produces a state with the same entropy, 
not necessarily preserves entropy when acting on other states.
In fact, such entropy preserving operations are in general not linear, since
they have to be \aw{constrained} to some input state. 
It was shown in Ref. \cite{Hulpke2006} that a quantum channel $\Lambda(\cdot)$ that
preserves entropy and, at the same time, respects
linearity, i.~e.\ $\Lambda(p \rho_1+ (1-p)\rho_2)=p\Lambda(\rho_1)+ (1-p)\Lambda(\rho_2)$,
has to be a unitary.

However, the entropy preserving operations can be microscopically described by global unitaries
in the limit of many copies \cite{Sparaciari16}.
Given any two states $\rho$ and $\sigma$ with equal entropies
$S(\rho)=S(\sigma)$, there exists an additional system 
of $O(\sqrt{n\log n})$ ancillary qubits and a global unitary $U$ such that,
\aw{as $n\rightarrow\infty$,}
\begin{equation}\label{eq:entropy-pres-micro}
  \left\|\tr_{\textrm{anc}}\left(U \rho^{\otimes n}\otimes\eta U^{\dagger}\right)-\sigma^{\otimes n}\right\|_1 \rightarrow 0,
\end{equation}
where the partial trace is performed on the ancillary qubits. 
Here, $\|X\|_1\coloneqq \tr \sqrt{X^\dagger X}$ is the one-norm. 
As shown in \cite[Thm.~4]{Sparaciari16}, the reverse statement is 
also true. In other words, if two states are related as in Eq.~\eqref{eq:entropy-pres-micro}, 
then they also have equal entropies. 

It is important to restrict entropy preserving operations that are also energy preserving, as we shall see later. The energy and entropy preserving operations can also be implemented using a global energy preserving unitary in the many-copy limit. More explicitly, in \cite[Thm.~1]{Sparaciari16}, it is shown that having two states $\rho$ and $\sigma$ with equal entropies and energies, i.e. ($S(\rho)=S(\sigma)$ and $E(\rho)=E(\sigma)$), is equivalent to the existence of some energy preserving $U$ and an additional system $A$ with $O(\sqrt{n\log n})$ of ancillary qubits with Hamiltonian $\norm{H_A}\leqslant O(n^{2/3})$ in some state $\eta$, for which \eqref{eq:entropy-pres-micro} is fulfilled.
The operator norm $\norm{X}$ of a Hermitian operator $X$ \aw{is} the largest of its eigenvalues in absolute value.
Note that the amount of energy and entropy of the ancillary system per copy vanishes in the large $n$ limit.

We expect entropy preserving operations to be also implemented in other ways than taking the
limit of many copies.
For instance, in Refs.~\cite{Wilming2017,Mueller2017},
thermal operations are extended to a class of operations, in which a
catalyst is allowed to build up correlations with the system.
For these operations,
the standard Helmholtz free energy singles out as the monotone that 
establishes the possible transitions between states, in contrast to the case of strict thermal operations, in which all the R\'enyi $\alpha$-free energies are required.
This suggests that entropy preserving operations could also be implemented
with a single copy by means of a catalyst that can become correlated with the system. 
Further investigation in this direction is needed.

\aw{Now that} the entropy preserving operations have been motivated and introduced, 
let us classify the set of states of a system in different equivalence classes depending 
on their entropy. Thereby, we establish a hierarchy of states according to their information content. 

\begin{defi}[Entropic equivalence class]
Two states $\rho$ and $\sigma$ on any quantum system of dimension $d$ are equivalent 
and belong to the same entropic equivalence class if and only if
both have the same von Neumann entropy,
\begin{equation} 
\rho \sim \sigma \ \text{ iff } \ S(\rho)=S(\sigma) .
\end{equation}
\end{defi}
\aw{We call such two states, $\rho$ and $\sigma$ \emph{iso-entropic}.} 
Assuming that the system has some fixed Hamiltonian $H$, one can take as a representative element of every class the state that \emph{minimizes} the energy within it, i.e.,
\begin{equation}
\gamma(\rho)\coloneqq\argmin_{\sigma \, : \, S(\sigma)=S(\rho)} E(\sigma) ,
\end{equation}
where $E(\sigma)\coloneqq \tr H \sigma$ is the energy of the state $\sigma$.

The maximum-entropy principle \cite{Jaynes57a, Jaynes57b} identifies the thermal state as the state that 
maximizes the entropy for a given energy. Conversely, one can show that, for a given entropy, 
the thermal state also minimizes the energy. We refer to this complementary property as 
\emph{min-energy principle} \cite{Pusz78, Lenard78, Alicki13}. \aw{To be precise, this duality holds for all finite inverse temperatures, $\beta < \infty$, corresponding to energies
strictly above the ground state energy and entropies strictly larger then the ground state
entropy.}
Thus, the min-energy principle identifies thermal states as the representative elements of every class, 
which is 
\begin{equation}\label{eq:thermal-state}
\gamma(\rho) =  \frac{\e^{-\beta(\rho) H}}{\tr \left(\e^{-\beta(\rho) H}\right)} .
\end{equation}
The inverse temperature $\beta(\rho)$ is the parameter that labels the equivalence class, 
to which the state $\rho$ belongs. We denote $\beta(\rho)$ as the intrinsic inverse temperature associated to $\rho$. 

The state $\gamma(\rho)$ is often termed the completely passive (CP) state \cite{Alicki13}. Note, in this article, we denote the CP states, synonymously, with the $\gamma(H, \beta)$, $\gamma(\rho)$,  $\gamma(\beta)$ and $\gamma(T)$, where $H$ is the Hamiltonian of system in the state $\rho$, and $\beta=1/T$ is the inverse of the temperature $T$. The $\gamma(H, \beta)$ is exclusively used in the cases where the Hamiltonian $H$ of the system is fixed. Moreover, the baths, in thermal equilibrium at certain temperature, are also CP states, and we sometime denote these thermal baths with $\gamma(T)$ and $\gamma(\beta)$.

The \CP state of the form $\gamma(H_S,\beta_S)$ has the following interesting properties \cite{Pusz78, Lenard78}:
\emph{
\begin{itemize}
\item[(P1)]{For a given entropy, it minimizes the energy.}
\item[(P2)]{\aw{Energy (entropy) monotonically increases (decrease) with the decrease (increase) in $\beta_S$.}}
\item[(P3)]{For non-interacting Hamiltonians \aw{and identical $\beta_T$}, $H_T=\sum_{X=1}^N \mathbb{I}^{\otimes X-1} \otimes H_X \otimes \mathbb{I}^{\otimes N-X}$, the joint complete passive state is \aw{the} tensor product of \aw{the} individual ones, i.e., $\gamma(H_T, \beta_T)=\otimes_{X=1}^N \gamma(H_X, \beta_T)$. \cite{Pusz78, Lenard78}.}
\end{itemize}
}


\section{Bound and free energies}
\label{sec:bound+free}
Let us now identify two relevant forms of internal energy: the free and the bound energy. 
The bound energy is \aw{motivated as} the amount of internal energy that cannot be accessed 
in the form of work. 
\aw{We emphasize once more} that it is a notion that \aw{depends} on the set of allowed operations. 
For the set of \EP operations, \aw{from} which the entropic classes and \CP states emerge, 
it is quantified \aw{as follows.}

\begin{defi}[Bound energy]
\label{defi:bound-energy}
For a state $\rho$ with the system Hamiltonian $H$, the \emph{bound energy} in it is
\begin{equation}\label{eq:bound-energy-def}
  B(\rho)\coloneqq  \min_{\sigma \, : \, S(\sigma)=S(\rho)} E(\sigma).
\end{equation}
By the previous discussion, $B(\rho)=E(\gamma(\rho))$,
where $\gamma(\rho)$ is the \CP state, with minimum energy, 
within the equivalence class to which $\rho$ belongs.
\end{defi}

Indeed, $B(\rho)$ is the amount of energy that cannot be extracted further, 
by exploiting any entropy preserving operations, as guaranteed by the min-energy principle. 
The above definition of bound energy also has strong connection with information 
content in the state. It can be easily seen that one could have access to this energy 
(in the form of work) only, if it allows an outflow of information from the system and vice versa.

In contrast to bound energy, free energy is the part of the internal energy 
that can be accessed with entropy preserving operations.

\begin{defi}[Free energy]
\label{defi:free-energy}
For a system $\rho$ with system Hamiltonian $H$, the \emph{free energy} stored in the system 
is given by
\begin{equation}\label{eq:free-energy-def}
  F(\rho) \coloneqq E(\rho)-B(\rho),
\end{equation}
where $B(\rho)$ is \aw{the} bound energy in $\rho$.
\end{defi}

Note that the free energy as defined in Eq.~\eqref{eq:free-energy-def} does not have a 
preferred temperature, unlike the standard \aw{out-of-equilibrium} Helmholtz free energy 
$F_T(\rho)\coloneqq E(\rho)-T\, S(\rho)$, where the temperature $T$ is decided beforehand 
with \aw{the} choice of a thermal bath. 
Nevertheless, our definition of free energy can be written in terms of both the relative 
entropy and the \aw{out-of-equilibrium} free energy as
\begin{equation}
  F(\rho) = T (\rho) D(\rho \| \gamma(\rho))
          = F_{T(\rho)}(\rho)-F_{T(\rho)}\left(\gamma(\rho)\right) ,
\end{equation}
for the intrinsic temperature $T(\rho)\coloneqq \beta(\rho)^{-1}$ that labels the equivalence class that contains $\rho$. 
Here the relative entropy is defined as  $D(\rho \| \sigma)=\tr \left( \rho \log \rho - \rho \log \sigma \right)$. 
Let us mention that the standard \aw{out-of-equilibrium} free energy is also denoted by 
$F_{\beta}(\rho)\coloneqq E(\rho)-\beta^{-1}\, S(\rho)$ in the rest of the manuscript, wherever we find it more convenient.  

The notions of bound and free energy, as we consider above, can be extended beyond 
single systems. If fact, in \aw{composite} systems they exhibit several interesting properties. 
For example, they can capture the presence and/or absence of inter-party correlations. 
To highlight these features, we consider a bipartite system below. 

\begin{lem}[Bound and free energy properties]
Given a bipartite system with non-interacting Hamiltonian 
$H_A\otimes \id + \id \otimes H_B$ in an arbitrary state $\rho_{AB}$ 
with marginals $\rho_{A/B}\coloneqq \tr_{B/A}(\rho_{AB})$. 
Then, the bound and the free energy \aw{satisfy} the following properties:
\begin{enumerate}
\item[(P4)] Bound energy and correlations:
\begin{equation}\label{eq:bound-energy-corr}
 B(\rho_{AB})\le B (\rho_A\otimes \rho_B) .
\end{equation}
\item[(P5)] Bound energy of composite systems:
\begin{equation}\label{eq:bound-energy-composite}
B (\rho_A\otimes \rho_B)\le B (\rho_A)+B (\rho_B) .
\end{equation}
\item[(P6)] Free energy and correlations:
\begin{equation}\label{eq:free-energy-corr}
F(\rho_A\otimes \rho_B)\le F(\rho_{AB}) .
\end{equation}
\item[(P7)] Free energy of composite systems:
\begin{equation}\label{eq:free-energy-composite}
F(\rho_A)+F(\rho_B)\le F(\rho_A\otimes \rho_B) .
\end{equation}
\end{enumerate}
\aw{Ineqs.}~\eqref{eq:bound-energy-corr} and \eqref{eq:free-energy-corr} \aw{are saturated with equality}
if and only if $A$ and $B$ are uncorrelated, i.e.~$\rho_{AB}=\rho_A\otimes\rho_B$. 
\aw{Ineqs.}~\eqref{eq:bound-energy-composite} and \eqref{eq:free-energy-composite} 
become equalities if and only if $\beta (\rho_A)=\beta(\rho_B)$.
\end{lem}

\begin{proof}
To prove {\it (P4)}, we have used the fact that, for a fixed Hamiltonian, the bound energy 
is monotonically increasing with \aw{the} entropy, i.e.
\begin{equation}
B(\rho) < B(\sigma) \text{ iff } S(\rho)< S(\sigma)\, ,
\end{equation}
together with the \aw{subadditivity} property $S(\rho_{AB})\le S (\rho_A\otimes \rho_B)$
\aw{of the von Neumann entropy}. 

The proof of {\it (P5)} relies on the definition of the bound energy, where one exploits entropy preserving 
operations to minimize \aw{the} energy. \aw{Namely, the iso-entropic equivalence relation is
stable under tensor product: if $\rho_A\sim\sigma_A$ and $\rho_B\sim\sigma_B$, then
$\rho_A\otimes\rho_B \sim \sigma_A\otimes\sigma_B$, by the additivity of the entropy.}
This immediately implies Eq.~\eqref{eq:bound-energy-composite}. 
To see that the inequality \eqref{eq:bound-energy-composite} is saturated iff the states \aw{have} 
identical intrinsic temperatures $\beta (\rho_A)=\beta(\rho_B)$,
we appeal to Definition \ref{def:equilibrium} and Lemma \ref{lm:betaAB}.

The other properties can be easily proven by noting that the total internal energy is 
sum of \aw{local} ones: $E(\rho_{AB})=E(\rho_A)+E(\rho_B)$, irrespective of inter-system correlations. 
\aw{Then, with the definition of free energy $F(\rho_X)=E(\rho_X)-B(\rho_X)$, and properties 
{\it (P4)} and {\it (P5)}, we easily arrive at {\it (P6)} and {\it (P7)}.}
\end{proof}

The above properties allow us to give an additional operational meaning to the free energy $F(\rho)$.
Let us \aw{recall} first that the work that can be extracted from a system in a state $\rho$ 
when having at our disposal an infinitely large bath at inverse temperature $\beta$ 
and global entropy preserving operations, is given by
\begin{equation}\label{eq:work-extracted-infinite-bath}
  W = F_\beta (\rho) - F_\beta (\gamma(\beta)) ,
\end{equation}
where $F_\beta (\rho)=E(\rho)-\beta^{-1}S(\rho)$ is the standard free energy with respect to 
the inverse temperature $\beta$ and 
$\gamma(\beta)$ is the thermal state.

\begin{lem}[Free energy vs. $\beta$-free energy]
\label{lem:free-energy-beta}
The free energy $F(\rho)$ corresponds to work extracted by using a bath at
the ``worst'' possible temperature:
\begin{equation}
  F(\rho)=\min_\beta \left(F_\beta (\rho) - F_\beta (\gamma(\beta))\right) .
\end{equation}
\aw{In words, the free energy of $\rho$ is the minimum extractable work
with respect to arbitrary temperature baths.}
The inverse temperature that achieves the minimum is the inverse
intrinsic temperature $\beta(\rho)$.
\end{lem}
\begin{proof}
\aw{It follows from} Lemma \ref{lem:work-extraction} together with (P7).
\end{proof}

\section{Energy-entropy diagram}
\label{sec:diagram}
Let us describe here the \emph{energy-entropy diagram} which appears often in \aw{the}
thermodynamics literature and in particular has recently been exploited in Ref.~\cite{Sparaciari16}.
Given a system described by a Hamiltonian $H$, a state $\rho$ is represented in the energy-entropy diagram by a point with coordinates $x_\rho\coloneqq(E(\rho), S(\rho))$, as shown in Fig.~\ref{fig:energy-entropy-diagram}. 
All physical states reside in a region that is lower bounded by the horizontal axis (i.e.~$S=0$), 
corresponding to the pure states, and upper bounded by the convex curve $(E(\beta),S(\beta))$, 
which represents the thermal states of both positive and negative temperatures. 
Let us denote such a curve as the thermal boundary. The inverse temperature associated to 
one point of the thermal boundary is given by the slope of the tangent line in such a point,
\begin{equation}
  \beta = \frac{\dd S(\beta)}{\dd E(\beta)} .
\end{equation}

\begin{figure}[h]
\includegraphics[width=\columnwidth]{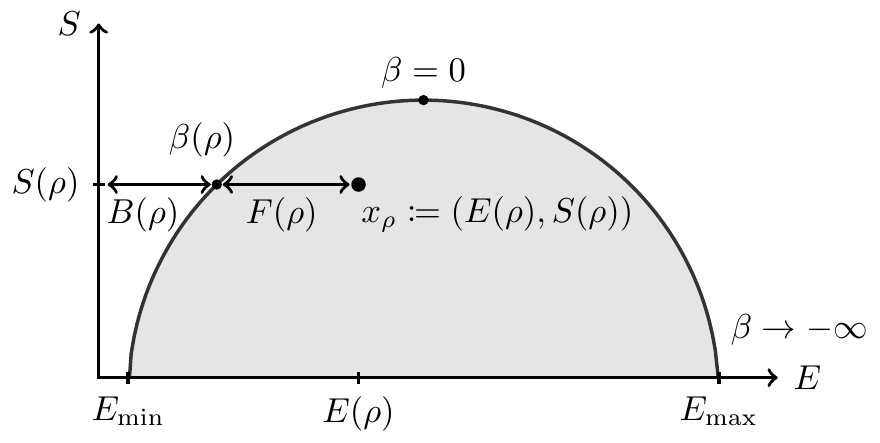}
\caption{Energy-entropy diagram. Any quantum state $\rho$ is represented in the diagram
as a point with coordinates $x_\rho\coloneqq(E(\rho), S(\rho))$. 
The free energy $F(\rho)$ is the distance in the horizontal direction from the thermal boundary. The bound energy $B(\rho)$ is the distance in the horizontal direction between the thermal boundary and the energy reference.}
\label{fig:energy-entropy-diagram}
\end{figure}

In Fig.~\ref{fig:energy-entropy-diagram}, the free energy and the bound energy are plotted given a state $\rho$. 
Its free energy $F(\rho)$ can be seen from the diagram as the horizontal distance from the thermal boundary. 
This is the part of \aw{the} internal energy which can be extracted without altering system's entropy. 
The slope of the tangent line of the thermal boundary in that point is the \aw{inverse} intrinsic temperature
$\beta(\rho)$, of the state $\rho$. The bound energy $B(\rho)$ is the distance in the horizontal direction between the thermal boundary and the energy reference and it can \aw{in} no way be extracted with entropy preserving operations.

Note that in general a point in the energy-entropy diagram represents multiple states, 
\aw{in fact, an entire equivalence class,} 
since different quantum states can have identical entropy and energy. 
As it is pointed out in \cite{Sparaciari16}, 
the energy-entropy diagram establishes a link between the microscopic and the macroscopic 
thermodynamics, in the sense that in the macroscopic limit of many copies all the thermodynamic
quantities only rely on the energy and entropy per particle.
More precisely, all the states with same entropy and energy are thermodynamically 
equivalent: \aw{as was shown in \cite{Sparaciari16}, they can be transformed 
into each other using energy-conserving unitaries 
in the limit of many copies $n\to\infty$, and using an ancillary system,  
which is of \aw{sublinear} size $O(\sqrt{n \log n})$ and with a Hamiltonian 
with its Hamiltonian $O(n^{2/3})$ sublinearly bounded}.

\section{Equilibrium and zeroth law}
\label{sec:zeroth}
The zeroth law of thermodynamics establishes \aw{the temperature as an intensive quantity, and its} the absolute scaling, in terms of mutual thermal equilibrium. It says that if a system $A$ is in thermal equilibrium with $B$ and again, $B$ is in thermal equilibrium with $C$, then $A$ will also be in thermal equilibrium with $C$. All the systems that are in mutual thermal equilibrium can be classified 
\aw{in a thermodynamical equivalence class}, and each class can be labelled with a unique parameter called temperature. In other words, at thermal equilibrium, temperatures of the individual systems will be equal to each other and also to their arbitrary collections, as there would be no spontaneous heat or energy flow in between. 
When a non-thermal state is brought in contact with a large thermal bath, then the system tends to acquire a thermally equilibrium state with the temperature of the bath. In this equilibration process, the system \aw{can} exchange both energy and entropy with \aw{the} bath and thereby \aw{minimize} its Helmholtz free energy. 

However, in the new setup, where a system cannot have access to \aw{asymptotically} large thermal baths, or in the complete absence of a bath, the equilibration process is expected to be considerably different. In the following, we introduce a formal definition of \aw{equilibrium}, based on information preservation and intrinsic temperature. 
\begin{defi}[Equilibrium and zeroth law]
\label{def:equilibrium}
Given a collection of systems $A_1, \ldots, A_n$ with non-interacting Hamiltonians $H_1, \ldots, H_n$ 
in a joint state $\rho_{A_1 \ldots A_n}$, \aw{we say that they are} mutually at equilibrium if and only if 
they jointly minimize the free energy as defined in \eqref{eq:free-energy-def}:
\begin{equation}
  F(\rho_{A_1 \ldots A_n})=0 .
\end{equation}
\end{defi}

Let us consider \aw{two systems A and B} in the states $\rho_A$ and $\rho_B$, with corresponding Hamiltonians 
$H_A$ and $H_B$, respectively. The equilibrium is achieved when they jointly attain an \aw{iso-entropic} state with minimum energy. 
The corresponding equilibrium state is then a completely passive state $\gamma(H_{AB},\beta_{AB})$ 
with the joint Hamiltonian $H_{AB}=H_A\otimes \mathbb{I} + \mathbb{I} \otimes H_B$. 
Further, the joint \CP state, following {\it (P3)},
is 
\begin{align}
  \gamma(H_{AB},\beta_{AB}) = \gamma(H_A, \beta_{AB}) \otimes \gamma(H_B, \beta_{AB}),
\end{align}
where the local systems assume \aw{their respective} completely passive states, \aw{importantly} with \aw{the} same $\beta_{AB}$. 
Recall that $\gamma(H_X, \beta_Y)=e^{-\beta_Y H_X}/\tr(e^{-\beta_Y H_X})$. Note, for interaction systems, i.e.  $H_{AB}\neq H_A\otimes \mathbb{I} + \mathbb{I} \otimes H_B$, one could minimize global energy using global entropy preserving operation too. This minimum energy state will give rise to a global equilibrium (\CP) state with intrinsic global temperature. However, the reduced state, after tracing out one sub-system, may not be a \CP state.

\aw{Another consequence of the definition is} that if two arbitrary \CP states $\gamma(H_A, \beta_A)$ and 
$\gamma(H_B, \beta_B)$ have $\beta_A \neq \beta_B$, then their combined equilibrium state 
\aw{$\gamma(H_{AB},\beta_{AB}) = \gamma(H_A, \beta_{AB}) \otimes \gamma(H_B, \beta_{AB})$ 
can still have a smaller bound energy without altering the total information content.} 
Therefore, the \aw{associated} joint \CP state acquires a unique $\beta_{AB}$ \aw{such that}
\begin{align}
\label{eq:eqlbAB}
  E(\gamma(H_A, \beta_A)) + E(\gamma(H_B, \beta_B)) > E(\gamma(H_{AB}, \beta_{AB})), 
\end{align}
which is \aw{follows immediately} from \aw{the} min-energy principle and {\it (P5)}. 
Moreover, if $\beta_A \geq \beta_B$, then Eq.~\eqref{eq:eqlbAB} dictates a bound on $\beta_{AB}$, 
as expressed in \aw{the following lemma.}

\begin{lem}
\label{lm:betaAB}
Any \aw{iso-entropic} equilibration process between $\gamma(H_A, \beta_A)$ and $\gamma(H_B, \beta_B)$, 
with $\beta_A \geq \beta_B$,
leads to a state $\gamma(H_{AB}, \beta_{AB})$ \aw{of mutual equilibrium}, where $\beta_{AB}$ satisfies
\begin{align}
 \beta_A \geq \beta_{AB} \geq \beta_B.
\end{align}
\end{lem}

\begin{proof}
For $\beta_A = \beta_B$,  {\it (P3)} and {\it (P5)} immediately lead to $\beta_A = \beta_{AB} = \beta_B$. 
What we \aw{want} to show is that, for $\beta_A > \beta_B$, the equilibration leads to $\beta_A > \beta_{AB} > \beta_B$. 
\aw{This} can be seen from \aw{the} information preservation condition on the equilibration process. 
\aw{Assume the contrary, say that after the equilibration the final inverse temperature is 
$\beta_{AB} \geq \beta_A > \beta_B$. But according to {\it (P2)}, the entropy is monotonically
decreasing with $\beta$, hence we would conclude that the total entropy of the system has decreased,
contradiction to the assumption of an iso-entropic process.
Likewise, $\beta_A > \beta_B \geq \beta_{AB}$ leads to the contradiction that the total entropy
of the system would have increased.
Therefore, the only possibility remaining is $\beta_A > \beta_{AB} > \beta_B$.}
\end{proof}

%

Now\aw{, with a} clear notion of equilibration and equilibrium state, 
which has minimum internal energy for a fixed information content, 
we \aw{can} restate \aw{the} \emph{zeroth law} in terms of \aw{the} intrinsic temperature.
By definition, the global \CP state has minimum internal energy and one cannot extract free energy by 
using global \EP operations. Further, a global \CP state not only assures that the individual states 
are also \CP but also \aw{enforces} that they share the same intrinsic temperature $\beta$ and 
\aw{has vanishing} inter-system correlations. 

We \aw{can} recover the traditional notion of thermal equilibrium as well as zeroth law. 
In traditional thermodynamics, the thermal bath is considerably \aw{much larger than the system under consideration}, 
and \aw{is initially in a} \CP state with a predefined temperature. 
\aw{With respect to the bath} Hamiltonian $H_B$, it can be expressed as 
$\gamma_B=e^{-\beta_B H_B}/\tr(e^{-\beta_B H_B})$. Now, if the state $\rho_S$ of the system $S$ (with $|S|\ll |B|$)
with Hamiltonian $H_S$ is brought in contact with \aw{the} thermal bath, then the global thermal equilibrium state 
will be a \CP state, i.e.~$\gamma_B \otimes \rho_S \xrightarrow{\Lambda^{ep}} \gamma_B^\prime \otimes \gamma_S^\prime$, 
with a global inverse equilibrium temperature $\beta_e$. With $|B| \gg |S|$, one \aw{now easily sees} that 
\aw{$\beta_e \approx \beta_B$ and $\gamma_S^\prime \approx \gamma(H_S, \beta_B) = e^{-\beta_B H_S}/\tr(e^{-\beta_B H_S})$.}

\section{Max-entropy principle \protect\\ vs.~min-energy principle}
\label{sec:athermality}
Let us mention that our framework, which is based on \ar{the principle of information conservation, is 
suited for the work extraction problem}. 
It assumes that the system has mechanisms to release energy to some battery or classical field, but it cannot interchange entropy with it. This contrasts with other situations in spontaneous equilibration, in which the system is assumed to evolve keeping constant the conserved quantities. In such scenario, the equilibrium state is described by the principle of maximum entropy, that is, the state the maximizes the entropy given the conserved quantities.

In the context of the maximum entropy principle, let us introduce the absolute athermality as
\aw{(cf.~\cite{Sparaciari16})}:
\begin{defi}[Athermality]
\label{defi:athermality}
The \emph{athermality} of a system in a state $\rho$ and Hamiltonian $H$ is 
the amount of entropy (information) that the system can still accommodate 
without increasing its energy, i.e.
\begin{equation}\label{eq:athermality-def}
  A(\rho) \coloneqq \max_{\sigma \, : \, E(\sigma)=E(\rho)} S(\sigma)-S(\rho),
\end{equation}
where the maximization is over all states $\sigma$ with energy $E(\rho)$.
\end{defi}
The state $\sigma$ that maximizes \eqref{eq:athermality-def} is obviously a thermal state.
In particular, it is the equilibrium state of an equilibration process guided by the maximum entropy principle,
\aw{where the energy of the system remains constant}. 
Let us call the temperature of such thermal state the \emph{spontaneous equilibration temperature},
and denote it by $\tilde\beta(\rho)$.
The athermality corresponds to the amount of entropy produced during such an equilibration
process.

Note that the spontaneous equilibration temperature $\tilde\beta(\rho)$ always differs
from the intrinsic temperature $\beta(\rho)$ unless $\rho$ is thermal.
In other words, the maximum entropy and minimum energy principles 
lead to different equilibrium temperatures.
As both entropy and energy are monotonically increasing functions of temperature, 
the equilibrium temperature determined by the minimum energy principle is always smaller than 
the one determined by the maximum entropy principle. 
This is represented in the energy entropy diagram of Fig.~\ref{fig:free-energy-athermality}, 
where the athermality can be seen as 
the vertical distance from a state $\rho$ to the thermal boundary. 
The point in which the thermal boundary intersects with the vertical 
line $E=E(\rho)$ represents the equilibrium thermal state given by the max-entropy principle. 
Its temperature $\tilde\beta(\rho)$
corresponds to the slope of the tangent line of the thermal boundary at that point.

\begin{figure}[h]
\includegraphics[width=.9\columnwidth]{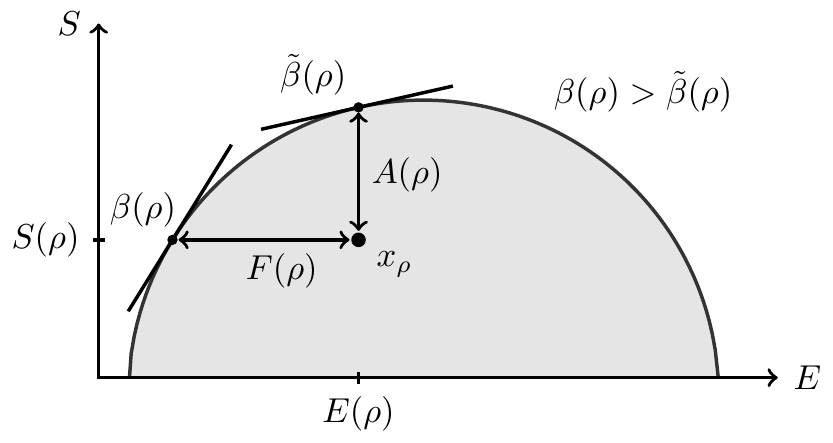}
\caption{Representation of the free energy $F(\rho)$ and the athermality $A(\rho)$
of a state $\rho$ in the energy-entropy diagram.
The intrinsic temperature $\beta(\rho)$ and the spontaneous equilibration
temperature $\tilde \beta(\rho)$ \aw{satisfy} $\beta(\rho)>\tilde \beta(\rho)$
due to the concavity of the thermal boundary.}
\label{fig:free-energy-athermality}
\end{figure}

We have chosen the name ``athermality'' for the quantity defined in \eqref{defi:athermality}
in agreement with Ref.~\cite{Sparaciari16}.
There, a quantity called $\beta$-athermality \aw{is} defined by
\begin{equation}\label{eq:beta-athermality-def}
  A_\beta(\rho) \coloneqq \beta E(\rho) - S(\rho) + \log Z_\beta ,
\end{equation}
with $Z_\beta=\tr(\e^{-\beta H})$ the partition function.
\aw{It is} introduced to characterise the energy-entropy diagram as the set of all 
points with positive entropy $S(\rho)\geq 0$ and positive $\beta$-athermality $A_\beta(\rho)\geq 0$ 
for all $\beta \in \mathbb{R}$.

\begin{lem}[Athermality vs. $\beta$-athermality]
\label{lem:athermality}
The athermality of a state $\rho$ can be written in terms of 
the $\beta$-athermalities as
\begin{equation}\label{eq:athermality-min-beta}
  A(\rho)= \min_{\beta} A_\beta (\rho) ,
\end{equation}
where the minimum is attained by the spontaneous equilibration temperature $\tilde\beta(\rho)$.
\end{lem}
\begin{proof}
We prove the lemma by giving a geometric interpretation to the $\beta$-athermality in the energy-entropy diagram.
In Fig.~\ref{fig:beta-free-energy-and-athermality}, 
let us consider the tangent line to the thermal boundary with slope $\beta$.
Such \aw{a} straight line is formed \aw{by} the set of all points \aw{satisfying}
\begin{equation}
  S-S_\beta =  \beta (E-E_\beta) ,
\end{equation}
where $S_\beta$ and $E_\beta$ are respectively the entropy and the energy of
the point \aw{on} the thermal boundary that \aw{lies on} the line.
The vertical distance (difference in entropy) of \aw{this} line from a point with coordinates
$(E(\rho),S(\rho))$ reads
\begin{equation}
\begin{split}
  S_*-S(\rho) &=\beta (E(\rho)-E_\beta) - S(\rho) \\
              &= \beta E(\rho) -S(\rho)-\beta (E_\beta - T S_\beta),
\end{split}
\end{equation}
where $S_*$ is the entropy of the line at $E=E(\rho)$ and $T=\beta^{-1}$.
By noticing that $\log Z_\beta= \beta (E_\beta - T S_\beta)$, 
we get that the $\beta$-athermality \aw{$A_{\beta}(\rho)$ of $\rho$} is the vertical distance of $x_\rho$ from
the line with slope $\beta$ tangent to the thermal boundary.
Eq.~\eqref{eq:athermality-min-beta} trivially follows from this.
\end{proof}

\begin{figure}[h]
\includegraphics[width=.9\columnwidth]{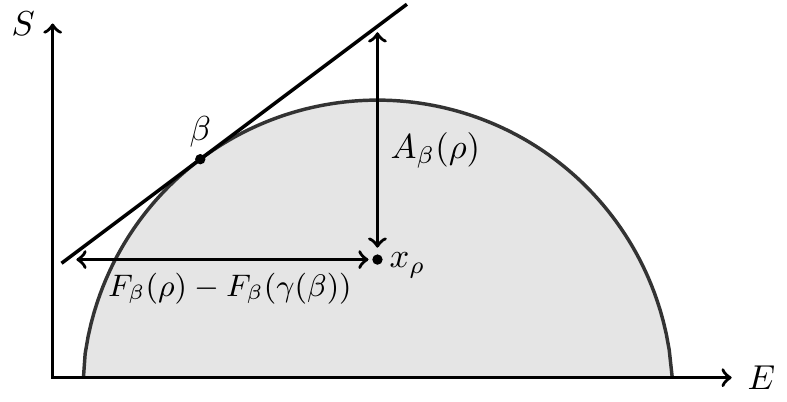}
\caption{Representation of the $\beta$-free energy difference $F_\beta(\rho)-F_\beta (\gamma(\beta))$ and the $\beta$-athermality $A_\beta(\rho)$ for a state $\rho$ with coordinates $x_\rho$  in the energy-entropy diagram.}
\label{fig:beta-free-energy-and-athermality}
\end{figure}

By identifying the $\beta$-free energy in Eq.~\eqref{eq:beta-athermality-def}, the $\beta$-athermality can be written as
\begin{equation}
  \label{eq:beta-free-energy-and-athermality}
  A_\beta(\rho)=\beta\left( F_\beta(\rho)-F_\beta (\gamma(\beta))\right) ,
\end{equation} 
that is, it is $\beta$ times the work that can be potentially
extracted from a state $\rho$ when having an infinite bath uncorrelated from the system 
at temperature $\beta$ (see Eq.~\eqref{eq:work-extracted-infinite-bath}).
From Eq.~\eqref{eq:beta-free-energy-and-athermality}, we see that the difference in 
$\beta$-free energies $F_\beta(\rho)-F_\beta (\gamma(\beta))$ can be represented 
in the energy-entropy diagram as the horizontal distance from the point $x_\rho$ 
to the tangent line with slope $\beta$. 
This is represented in Fig.~\ref{fig:beta-free-energy-and-athermality}.
Hence, given an infinitely large bath at some inverse
temperature $\beta$, all the states with the same work potential
lie on a line with slope $\beta$ in the energy-entropy diagram. 
The work extracted from a state $\rho$ can be decomposed into 
its free energy $F(\rho)$, the part of energy that could have been extracted without bath, 
and the rest, \aw{i.e.~the part of the energy} that has been accessed thanks to the bath. 
This latter part is closely related to heat. The definitions of work and heat are discussed in the next section.

Furthermore, thanks to Lemmas \ref{lem:free-energy-beta} and \ref{lem:athermality}, 
both the free energy and the athermality can be expressed in terms of a minimization over standard $\beta$-free energies
\begin{equation}
\begin{split}
  F(\rho)&=\min_\beta \left(F_\beta (\rho) - F_\beta (\gamma(\beta))\right), \\
  A(\rho)&=\min_\beta \left( \beta \left( F_\beta(\rho)-F_\beta (\gamma(\beta))\right)\right),
\end{split}
\end{equation}
where the minimum is attained by the intrinsic temperature $\beta(\rho)$ 
and the spontaneous equilibration temperature $\tilde \beta(\rho)$, \aw{respectively}.

Note that, while for positive temperatures both the athermality and the free energy are a 
measure of out of equilibrium, for negative temperatures, states with small athermality are 
highly active and have huge free energies.

Let us finally show that in the equilibration of a hot body in contact with a cold one, 
the intrinsic temperature and the spontaneous equilibration one are not so much different. 
In Fig.~\ref{fig:fig:equilibrium-temperature}, we \aw{sketch the} energy-entropy diagram for the particular case
of the equilibration of a system composed of two identical subsystems initially at different temperatures 
$\beta_A$ and $\beta_B$. As expected and due to the concavity of the thermal boundary, 
the equilibrium temperature $\beta_E^{-1}$ given by the constant energy \aw{constraint} (max-entropy principle) is larger
that the one given by the constant entropy constrain (min-energy principle), $\beta_S^{-1}$,
i.e.~$\beta_E < \beta_S$.
The difference in bound energies of these two thermal states corresponds precisely to the work
extracted in the constant entropy scenario.

\begin{figure}
\includegraphics[width=0.9\columnwidth]{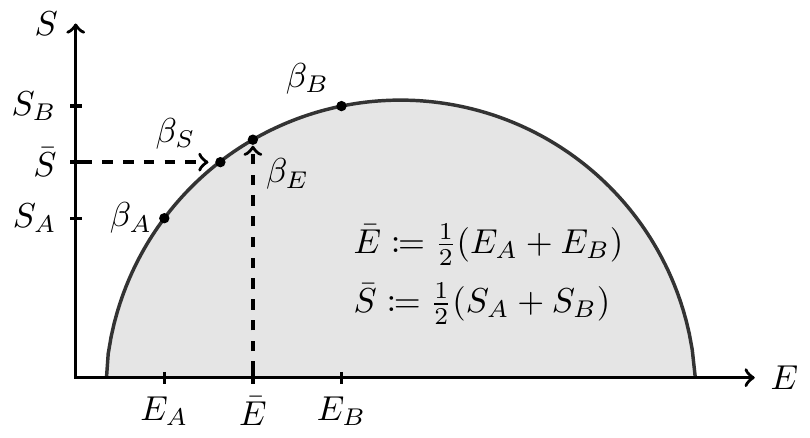}
\caption{Energy-entropy diagram of a system with Hamiltonian $H$.
Two systems with the same Hamiltonian $H$ and initially at different temperatures $\beta_A$ and $\beta_B$ equilibrate to different temperature depending on the approach taken: minimum energy principle vs. maximum entropy principle. The equilibrium temperature when entropy is conserved $\beta_S^{-1}$ is always smaller than the equilibrium temperature when energy is conserved $\beta_E^{-1}$, i.~e. $\beta_S > \beta_E$.}
\label{fig:fig:equilibrium-temperature}
\end{figure}

\section{Work, heat and the first law}
\label{sec:first}
In thermodynamics, the first law deals with the conservation of energy. 
In addition, it dictates the distribution of energy over work and heat, 
that are the two forms of thermodynamically relevant energy. 

Let us define a thermodynamic process involving a system $A$ and a bath $B$ as 
a transformation $\rho_{AB}\to \rho_{AB}'$ that conserves the global entropy $S(\rho_{AB})=S(\rho_{AB}')$. 
In standard thermodynamics, the bath is assumed to be initially in a thermal state and completely uncorrelated 
from the system. In such \aw{a} scenario, the heat dissipated in \aw{the} process is usually defined as 
the change in the internal energy of the bath,
$\Delta Q=E(\rho_B')-E(\rho_B)$, where $\rho_B^{(\prime)}=\tr_A \rho_{AB}^{(\prime)}$ 
is the reduced state of the bath. \aw{Here, and in the following, primed quantities refer to the
state after the process, whereas their unprimed versions refer to the state before the process.}
This definition, however, may have some limitations and alternative definitions have been discussed 
recently \cite{Bera16}. In particular \aw{the information theoretic approach from \cite{Bera16}} 
suggests \aw{that heat be} defined as $\Delta Q=T_B \Delta S_B$ with $T_B$ being the temperature 
and $\Delta S_B= S(\rho_B')-S(\rho_B)$ the change in von Neumann entropy of the bath. 
Note that $\Delta S_B=-\Delta S(A|B)$ is also the conditional entropy change in system $A$, 
conditioned on the bath $B$, where $S(A|B)=S(\rho_{AB}) - S(\rho_B)$, 
which can be understood as the information flow from the system to the bath in the presence of correlations.

In the present approach, we go beyond the restriction that the environment has to have a 
definite predefined temperature, and be in \aw{a} state of the Boltzmann-Gibbs form. 
For an arbitrary environment, which could even be athermal, \aw{we propose to quantify the heat transfer}
in terms of bound energies, \aw{as follows}.
\begin{defi}[Heat]
\label{def:heat}
Given a system $A$ and its environment $B$, the heat dissipated by the system $A$ in the 
process $\rho_{AB} \xrightarrow{\Lambda^{ep}} \rho_{AB}'$ is defined as the change in the 
bound energy of the environment $B$, i.e.
\begin{equation}\label{eq:heat-def}
  \Delta Q \coloneqq B(\rho_B')-B(\rho_B),
\end{equation}
\aw{where $B(\rho_B)$ and $B(\rho_B^{(\prime)})$ are the initial and final bound energy of the bath $B$,
respectively.}
\end{defi}

Note that heat is a \aw{somewhat} process dependent quantity from the point of view of the work system $A$, 
in the sense that there \aw{can be different} processes with the same initial and final marginal state for $A$, 
but different marginals for $B$.
Because the global process is entropy preserving, processes that lead to the same marginal for $A$, 
but different marginals with different entropies for $B$, necessarily imply a different amount of correlations 
between $A$ and $B$, as measured by the mutual information $I(A:B)\coloneqq S_A+S_B-S_{AB}$. 
\aw{Namely,} $\Delta S_A + \Delta S_B = \Delta I (A:B)$.

\begin{figure}
\includegraphics[width=0.8\columnwidth]{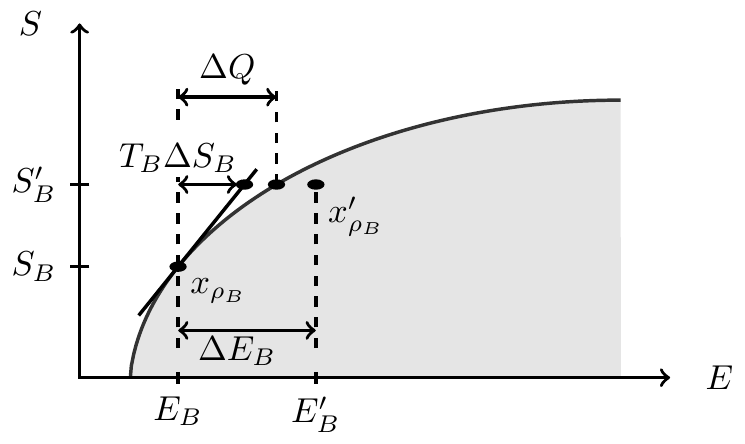}
\caption{Representation in the energy-entropy diagram of the different notions of heat: $\Delta Q$ as change of the bound energy of the bath, $\Delta E_B$ as the change in internal energy, and $T_B \Delta S_B$. In this example, the bath is initially thermal but this does not have to be necessarily the case.}
\label{fig:heat-definitions}
\end{figure}

\begin{lem}[Connections among \aw{definitions of heat}]
\label{lem:connections-heat}
Given a system $A$ and its environment $B$, the heat \aw{disspated} by $A$ 
in a globally entropy preserving process $\rho_{AB} \xrightarrow{\Lambda^{ep}} \rho_{AB}'$ can be written as
\begin{equation}\label{eq:heat-integral}
\Delta Q=\int_{S_B}^{S_B'}T(s) d s,
\end{equation}
where $S_B^{(\prime)}=S(\rho_B^{(\prime)})$ is the initial (final) entropy of $B$, 
and $T(S)$ is the intrinsic temperature of the thermal state of $B$ with entropy $S$. 
Because of the continuity of $T(S)$, there exists \aw{hence an} 
$s_*\in [S_B,S_B']$ such that
\begin{equation}\label{eq:mean-value}
  \Delta Q = T(s_*) \Delta S_B\, .
\end{equation}
Thus, \aw{Definition} \ref{def:heat} becomes $\Delta Q = T \Delta S_B$ in the limit of small entropy changes
\aw{along the process, in which case $T(S_B)\approx T(S_B') \approx T(s_*)$}, 
while it differs from the common definition in general, and $\Delta Q\ne \Delta E_B$.

In the case that the initial state of the environment is thermal 
$\rho_B \aw{= \gamma(H_B,T^{-1})} = \e^{-H_B/T}/\tr(\e^{-H_B/T})$ 
with Hamiltonian $H_B$ at temperature $T$, the heat is lower and upper bounded by
\begin{equation}\label{eq:heat-defs-bounds}
T \Delta S_B\le \Delta Q \le \Delta E_B ,
\end{equation}
where the three quantities have been represented in the energy-entropy diagram
in Fig.~\ref{fig:heat-definitions}.
If, \aw{in addition}, the process perturbs the environment \aw{only slightly}, $\rho_B'=\rho_B+\delta \rho_B$, 
then the three definitions of heat coincide, in the sense that
\begin{equation}\label{eq:heat-dif-def}
 T \Delta S_B +O(\delta \rho_B^2)=\Delta Q=\Delta E_B-O(\delta\rho_B^2) .
\end{equation}
That is, in the limit of large thermal baths, the three definitions become equivalent.
\end{lem}
\begin{proof}
Eq.~\eqref{eq:heat-integral} is a consequence of the definitions of heat \eqref{eq:heat-def} 
and bound energy \eqref{eq:bound-energy-def}, together with Eq.~\eqref{eq:thermal-state}. 
\aw{Next, Eq.~\eqref{eq:mean-value}} is proven by using the mean value theorem for the function $T(s)$.
The lower bound in \eqref{eq:heat-defs-bounds} is due to the concavity of the thermal boundary together
with the reminder theorem of a first order Taylor expansion of the thermal boundary (see Fig.~\ref{fig:heat-definitions}). 
The upper bound in \eqref{eq:heat-defs-bounds}
is a consequence of the positivity of the free-energy.
Finally, Eq.~\eqref{eq:heat-dif-def} follows from the fact that the thermal state $\rho_B$ is a minimum of the 
standard free energy $F_{T}(\rho)=E(\rho)-T S(\rho)$, which implies $\Delta F_{T,B}=O(\delta \rho_B^2)$,
and $\Delta E_B=T \Delta S_B+O(\delta\rho_B^2)$. 
\end{proof}

\aw{Now that heat has been unambiguously defined, we can pass to the definition of work (cf.~\cite{FS:heat-and-work}).}

\begin{defi}[Work]
\label{defi:work} 
For an arbitrary entropy preserving transformation involving a system $A$ and its environment $B$, 
$\rho_{AB} \rightarrow \rho_{AB}'$, with fixed non-interacting Hamiltonians $H_A$ and $H_B$, 
the worked performed on the system $A$ is defined as
\begin{equation}
  \Delta W_A \coloneqq W - \Delta F_B,
\end{equation}
where $W=\Delta E_A + \Delta E_B$ is the work cost of implementing the global transformation 
and $\Delta F_B= F(\rho_B')-F(\rho_B)$ \aw{is the change in free energy of the bath}.
\end{defi}

\aw{Equipped} with the notions of heat and work, the first law 
takes the form of a mathematical identity.

\begin{lem}[First law]
\label{lm:1stLaw} 
For an arbitrary entropy preserving transformation involving a system $A$ and its environment $B$, 
$\rho_{AB} \rightarrow \rho_{AB}'$, with fixed non-interacting Hamiltonians $H_A$ and $H_B$, 
the change in energy for $A$ is distributed as
\begin{equation}
  \Delta E_A = \Delta W_A - \Delta Q .
\end{equation}
\end{lem}
\begin{proof}
The proof follows from the definitions of work and heat. 
\end{proof}

\aw{The interpretation of this identity is that under a given process,} 
$-\Delta W_A$ is the amount of energy that can be transferred \aw{to} and stored in a battery. 
While the heat $\Delta Q_A$ is the change in energy due to flow of information from/into the system.

\section{Second law}
\label{sec:second}
The second law of thermodynamics is formulated in many different forms: \aw{as} an upper bound on the 
extracted work, the impossibility of converting heat into work completely, etc. 
In this section we show how all these formulations are a consequence of the principle of 
\aw{information} conservation. \aw{Note that the formulation of the second law following from global 
entropy conservation is not new (e.g. in \cite{Esposito10}), but it was thought to require 
a bath with well-defined temperature. 
What is new in the present approach is that we do not need even a bath, let alone that it is initially 
in a thermal state with pre-defined temperature.} 

\aw{One of the most important questions in thermodynamics, both historically and 
conceptually, is to which extent the energy in a heat bath can be converted into work. 
The second law puts a limit on the amount of work extracted.
Over the years, research has led to various formulations of 
second law in standard thermodynamics, like the Clausius statement, 
the Kelvin-Planck statement and the Carnot statement, to mention a few. 

Similar questions can be posed in the framework considered here, in terms of bound energy. 
In the following, we present analogous forms of second laws that consider these questions 
qualitatively as well as quantitatively.}

\subsection{Work extraction}

\ar{One of the major concerns in thermodynamics} is the \aw{conversion of any} 
form of energy into work, which can be used for any application with certainty. 

\begin{lem}[\aw{Second law: work extraction}]\label{lem:work-extraction}
For an arbitrary composite system $\rho$, the extractable work by 
any entropy preserving process $\rho \to \rho'$,
$W=E(\rho)-E(\rho')$ is upper-bounded by the free energy
\begin{equation}
  W \le F(\rho),
\end{equation}
\aw{which is saturated with equality} if and only if $\rho'=\gamma(\rho)$.

If the system \aw{initially is in the particular state} 
$\rho=\rho_A\otimes \gamma_B(T_B)$ with $\gamma_B(T_B)$ being thermal at temperature $T_B$, then
\begin{equation}\label{eq:recover-standard-free-energy}
  W \le F_{T_B}(\rho_A) - F_{T_B}\left(\gamma_A(T_B)\right)
\end{equation}
where $F_T(\cdot)$ is the standard \aw{out-of-equilibrium} free energy, and the \aw{inequality is saturated}
in the limit of an infinitely large bath (infinite heat capacity).
\end{lem}
\begin{proof}
By using $E(\rho')\geq B(\rho')$, one gets
\begin{equation}
W=E(\rho)-E(\rho')\le E(\rho)-B(\rho')=F(\rho),
\end{equation}
where note that $B(\rho')=B(\rho)$.

When the \aw{initial state has the} composite structure 
$\rho=\rho_A \otimes \gamma_B(T_B)$, we get
\begin{equation}
\label{eq:work-dif-free-energies}
\begin{split}
  W &\leq F(\rho_A\otimes \gamma_B(T_B)) \\
    &\leq F_T(\rho_A\otimes \gamma_B(T_B)) - F_T(\gamma_A(T_B) \otimes \gamma_B(T_B)) \\
    &=    F_T(\rho_A) - F_T(\gamma_A(T_B)),
\end{split}
\end{equation}
we have used Lemma \ref{lem:free-energy-beta}.

Finally, in the limit of an infinitely large bath, the
intrinsic temperature $T_{AB}$ tends to the bath temperature
$T_B$ and $\Delta Q = T_B \Delta S_B$, which saturates the bound in 
Eq.~\eqref{eq:work-dif-free-energies}.

To see that the intrinsic temperature $T_{AB}$ tends to $T_B$ in the
limit of large baths,
let us increase the bath size by adding several copies of it. 
The entropy change of such bath reads
\begin{equation}
\begin{split}
\Delta S_B&=S(\gamma_B^{\otimes n}(T_{AB}))-S(\gamma_B^{\otimes n}(T_{B}))\\
&=n \left(S(\gamma_B(T_{AB}))-S(\gamma_B(T_{B}))\right) .
\end{split}
\end{equation}
By considering the bound $\Delta S_B = -\Delta S_A\le |A|$, we get
\begin{equation}
S(\gamma_B(T_{AB}))-S(\gamma_B(T_{B})) \leq \frac{|A|}{n} ,
\end{equation}
which, together with the continuity of $S(\gamma(T))$, proves $\lim_{n\to \infty}T_{AB}=T_B$.
\end{proof}

\subsection{Clausius statement}
\ar{The second law of thermodynamics not only dictates the direction of state transformations, 
but also establishes a fundamental bound on the amount of work that can be extracted by them}. 
Here we first concentrate on the bounds on extractable work and introduce analogous versions 
of \aw{the} second law in our setup. 

\begin{lem}[\aw{Second law: Clausius statement}]
Any \aw{iso-entropic} process involving two bodies $A$ and $B$
in an arbitrary state, with intrinsic temperatures $T_A$ and $T_B$, respectively
\aw{satisfies} the following inequality:
\begin{equation}\label{eq:generalized_Clausius}
  (T_B -T_A) \Delta S_A \geqslant \Delta F_A + \Delta F_B + T_B \Delta I(A:B) - W ,
\end{equation}
where $\Delta F_{A/B}$ is the change in the free energy of the body $A/B$,
$\Delta I (A:B)$ is the change of mutual information, 
and $W=\Delta E_A + \Delta E_B$ is the amount of external work performed on the \aw{total system}.

In \aw{the} absence of initial correlations between \aw{the} two bodies $A$ and $B$,
the states being initially thermal, and no external work performed, this implies
\begin{equation}
  (T_B -T_A) \Delta S_A \geqslant 0\, ,
\end{equation}
meaning that no \aw{iso-entropic} equilibration process is possible whose {\bf sole} 
result is the transfer of heat from a cooler to a hotter body.
\end{lem}

\begin{proof}
The definition of free and bound energy implies that
\begin{equation}\label{eq:energy_balance}
W=\Delta F_A + \Delta F_B + \Delta Q_A + \Delta Q_B\, ,
\end{equation}
where have used the definition of heat as the change of bound energy of the environment.
From Eq.~\eqref{eq:heat-defs-bounds}, one gets,
\begin{equation}
\Delta Q_A +\Delta Q_B \geqslant T_B \Delta S_B + T_A \Delta S_A\, ,
\end{equation}
where $T_{A/B}$ is the initial intrinsic temperature of the body $A/B$.
Due to the conservation of the total entropy, the change in mutual information 
can be written as $\Delta I (A:B)=\Delta S_A+\Delta S_B$.
Putting everything together and noting that $\textrm{sign}(\Delta B)=\textrm{sign}(\Delta S)$ completes the proof.
\end{proof}

The terms \aw{on} the right hand side of Eq.~\eqref{eq:generalized_Clausius}
show the three reasons for which the standard Clausius statement can be violated. 
Either because of the process not being spontaneous (external work is performed, $W>0$), 
or due to initial states having free energy which is consumed, or the presence of 
correlations \cite{Bera16}. 

An alternative formulation of \aw{the} Clausius statement, 
for initial and final equilibrium states, is considered in Appendix \ref{App:C}.

\subsection{Kelvin-Planck statement}
While the Clausius statement tells us that heat cannot flow  spontaneously from a hotter to a colder body, 
the Kelvin-Plank formulation of second law states that, when heat going from a hotter to a colder body, 
it cannot be completely transformed into work.

\begin{lem}[\aw{Second law: Kelvin-Planck statement}]
Any \aw{iso-entropic} process involving two bodies $A$ and $B$
in an arbitrary state 
satisfies the following energy balance
\begin{equation}\label{eq:generalized_Kelvin-Planck}
  \Delta Q_B+\Delta Q_A = -(\Delta F_A + \Delta F_B) + W ,
\end{equation}
where $\Delta F_{A/B}$ is the change in the free energy of the body $A/B$, $\Delta Q_{A/B}$ 
the heat dissipated by the body $A/B$, and $W=\Delta E_A + \Delta E_B$ is the amount of external 
work performed on the \aw{total system}.

In the case of the reduced states being thermal, 
and for a work extracting process $W<0$, the above equality becomes
\begin{equation}\label{eq:Kelvin-Planck}
  \Delta Q_B+\Delta Q_A \leqslant  W < 0 .
\end{equation}
Finally, in \aw{the} absence of initial correlations, Eq.~\eqref{eq:Kelvin-Planck} implies that 
no \aw{iso-entropic} equilibration process is possible whose {\bf sole} result is the absorption 
of bound energy (heat) from an equilibrium state and its complete conversion into work.
\end{lem}

\begin{proof}
Eq.~\eqref{eq:generalized_Kelvin-Planck} is a consequence of the energy balance \eqref{eq:energy_balance}.
Eq.~\eqref{eq:Kelvin-Planck} follows from \eqref{eq:generalized_Kelvin-Planck} by considering reduced states that are initially thermal and thereby $\Delta F_{A/B}\geqslant 0$.
To prove the final statement is sufficient to notice that entropy preserving processes on initially uncorrelated systems fulfill $\Delta S_A + \Delta S_B \geqslant 0$, which together with \eqref{eq:Kelvin-Planck}
implies that $\textrm{sign}(\Delta Q_A) =-\textrm{sign}(\Delta Q_B)$.
\end{proof}

An alternative formulation of Kelvin-Planck statement, for initial and final equilibrium states, is considered
in Appendix \ref{App:KP}.

\subsection{Carnot statement}
A \aw{traditional} heat engine extracts work from a situation in which two \aw{thermal} baths 
$A$ and $B$ have different temperatures $T_A$ and $T_B$. The work extraction is usually implemented 
in practice by means of a working body $S$ which cyclically interacts with $A$ and $B$. Here we 
\aw{are not concerned with how precisely} the working medium interacts with the baths $A$ and $B$, 
but just assume that at the end of every cycle the working body is left in its initial state and uncorrelated 
with the two bath. In other words, the working body only has the role to move energy from one bath to the other. 
From this perspective, the \aw{heat engine can be analyzed} by studying
the changes of the environments $A$ and $B$.
In contrast to the standard situation, here we will not assume that \aw{the} 
baths are infinitely large, but that \aw{they} have a similar size as the system, 
and so the loss or gain of energy changes their (intrinsic) temperature. 

\aw{Without loss of generality we assume $T_A<T_B$ and $\rho_{AB}=\gamma_A(T_A)\otimes \gamma_B(T_B)$. 
After operating the machine for one or several complete cycles, the baths experience a change $\rho_{AB}\to \rho_{AB}'$.}

We define the efficiency of work extraction in a heat engine as the fraction of energy
that is taken from the hot bath which is transformed into work:
\begin{equation}
  \eta\coloneqq \frac{W}{-\Delta E_B},
\end{equation}
where $-\Delta E_B=E_B-E_B'>0$ is the energy taken from the hot bath, and $W$ is the work extracted.
In the following, we upper-bound the efficiency of any heat engine.

\begin{lem}[\aw{Second law: Carnot statement}]
For an engine working with two initially uncorrelated environments,
$\rho_{AB}=\gamma_A(T_A)\otimes \gamma_B(T_B)$, 
each in a local equilibrium state with intrinsic temperatures $T_B > T_A$, 
the efficiency of work extraction is bounded by
\begin{equation}
  \eta \leqslant 1 - \frac{\Delta B_A}{-\Delta B_B},
\end{equation}
where $\Delta B_A$ and $\Delta B_B$ are the \aw{changes} in bound energies of the systems $A$ and $B$, respectively.

In the limit of large baths and under global entropy preserving operations,
the Carnot efficiency is recovered,
\begin{equation}\label{eq:Carnot}
  \eta \leqslant 1 - \frac{T_A}{T_B} .
\end{equation}
\end{lem}

\begin{proof}
The maximum extractable work due to the transformation $\rho_{AB}\to \rho_{AB}'$ is given by
\begin{equation}
  W=F(\rho_{AB})-F(\rho_{AB}')=(-\Delta E_B) - \Delta E_A>0 .
\end{equation}
The efficiency then reads
\begin{equation}
  \eta = 1 - \frac{\Delta E_A}{-\Delta E_B} .
\end{equation}
The condition of $A$ being initially at equilibrium implies that
$\Delta F_A\geqslant 0$, from \aw{which it follows that} $\Delta E_A > \Delta B_A$, 
and analogously for $B$.
Thus, 
\begin{equation}
  \eta\leqslant 1 - \frac{\Delta B_A}{-\Delta B_B} .
\end{equation}
In the limit of large baths, in which $\Delta B_A \ll B_A$, 
the change in bound energy becomes $\Delta B_A=T_A \Delta S_A$.
Hence, 
\begin{equation}\label{eq:Carnot-step}
  \eta\leqslant 1 - \frac{T_A \Delta S_A}{-T_B \Delta S_B} .
\end{equation}
If the process is globally entropy preserving, i.e.~$S_{AB}'=S_A+S_B$,
then $\Delta S_A+\Delta S_B \geqslant 0$, or alternatively $\Delta S_A \geqslant -\Delta S_B$,
which together with \eqref{eq:Carnot-step} implies \eqref{eq:Carnot}.
\end{proof}

\section{\aw{Third law?}}
\label{sec:third-law}
The third law of thermodynamics establishes the impossibility of attaining the absolute zero temperature.
According to Nernst, ``it is impossible to reduce the entropy of a system to its absolute-zero value in 
a finite number of operations''.
The third law of thermodynamics has been very recently proven in a resource theoretic setting \cite{Masanes2017}.
The unattainability of absolute zero entropy is a consequence of the \aw{unitary} character of 
the transformations considered. For instance, consider the state transformation
\begin{equation}\label{eq:erasure}
    \rho_S\otimes \rho_B \longrightarrow \proj{0}\otimes\rho_B' ,
\end{equation}
where $\rho_B$ is \aw{the} thermal state of the bath $B$, \aw{which models the cooling down (erasure) of the system $S$}, 
initially in a state $\rho_S$, \aw{which we assume to have} $\textrm{rank}(\rho_S)>1$.
Here the dimension of the Hilbert space of the bath, $d_B$, is considered to be arbitrarily large but finite.
As $\rho_B$ is a full-rank state, the left hand side and the right hand side of Eq.~\eqref{eq:erasure} have 
different ranks, and they cannot be transformed via a unitary transformation.
Thus, irrespective of work supply, one cannot attain the absolute zero entropy state.

The zero entropy state can only be achieved for infinitely large baths and a sufficient work supply. 
Assuming a locality structure for the bath's Hamiltonian, this would take an infinitely long time.
However, in the finite dimensional case, a quantitative bound on the achievable temperature given a 
finite amount of resources (e.~g.\ work, time) is derived in \cite{Masanes2017}.

Note that our set of entropy preserving operations is more powerful than unitaries. 
\ar{In particular, transformation \eqref{eq:erasure} is always possible as long as 
the entropy of the system can be allocated in the bath, i.e.
\begin{equation}
  S(\rho_S)\leqslant \log d_B - S(\rho_B)\, .
\end{equation}
The required work to do so will be $W=F(\proj{0}\otimes\rho_B')-F(\rho_S\otimes \rho_B)$.}
As it has been discussed in Sec.~\ref{sec:EPoperations},
entropy preserving operations can be implemented
as global unitaries acting on \aw{asymptotically} many copies.
Thus, the attainability of the absolute zero entropy state by means of entropy preserving operations 
is in agreement with the cases of infinitely large baths and unitary operations \cite{Masanes2017}.

In conclusion, the third law of thermodynamics is a consequence of the microscopy reversibility 
(unitarity) of the transformation and is not respected by entropy preserving operations.

\section{A temperature-independent resource theory of thermodynamics}
\label{sec:resource}
In this section we connect the framework developed above with the \aw{recently developed ``standard''} 
approach to thermodynamics as a resource theory. The main ingredients of a resource theory are the 
state space, which is usually compatible with a composition operation ---for instance, quantum states 
together with the tensor product---, and a set of of allowed \aw{operations, which should
facilitate state transformations that are in general irreversible}. 
In a first attempt to have a resource theory of thermodynamics from our framework,  
\aw{it would seem that we are bound to transformations within an equivalence class, 
that is, between states on the same Hilbert space and with equal entropies, which is of course 
a very poor perspective}. 
To compare and explore \aw{the interconvertibility of} states with different entropies, 
it is necessary to extend the set of operations, and in particular to add a composition rule.
%
In order to be able to connect states with different entropies and/or number of particles it is natural to 
consider different number of copies, \aw{and ask about \emph{rates} of transformations.}

Let us restrict first to the case in which the set of operations is entropy preserving regardless \aw{of} the energy. 
We \aw{ask} what is the rate of transforming $\rho$ into $\sigma$ and consider \aw{first the case} that $S(\rho)\le S(\sigma)$. 
Then, \aw{for} large enough $n$ and $m$ we have
\begin{equation}
  \label{eq:equal-entropies}
  S(\rho^{\otimes n}) \sim S(\sigma^{\otimes m}).
\end{equation}
With such a trick of considering a different number of copies, we have brought $\rho$ and $\sigma$, 
which had different entropies, into the same manifold of equal entropy.
There is a subtlety here which is that $\rho^{\otimes n}$ and $\sigma^{\otimes m}$ live in general \aw{on different Hilbert spaces}. 
This can be easily circumvented by adding some product state, \aw{i.e.
\begin{equation}\label{eq:rho-to-sigma-times-0}
  \rho^{\otimes n} \longrightarrow \sigma^{\otimes m}\otimes \proj{0}^{\otimes (n-m)},
\end{equation}}
where $\rho^{\otimes n}$ and $\sigma^{\otimes m}\otimes \proj{0}^{\otimes (n-m)}$ \aw{now live on
the same Hilbert space}.
Eq.~\eqref{eq:rho-to-sigma-times-0} can also be understood as a process of \aw{information} compression,
in which the information in $n$ copies of $\rho$ is compressed to $m<n$ copies of $\sigma$,
and $n-m$ systems have been erased.
Thus, in the case of having only an entropy preserving constraint, the 
transformation rate can be determined from \eqref{eq:equal-entropies}
\begin{equation}
  \label{eq:rate-entropies}
  r \coloneqq \lim_{n\rightarrow \infty} \frac{m}{n} = \frac{S(\rho)}{S(\sigma)} .
\end{equation}

In thermodynamics, however, energy must also be taken into account. 
When one does that, it could well be that the process \eqref{eq:rho-to-sigma-times-0}
is not energetically \aw{allowed, since it might happen that $E(\rho^{\otimes n}) \neq E(\sigma^{\otimes m})$}.
In such a case, we would need more copies of $\rho$, in order for the transformation
to be energetically favorable. 
This would force us to add an entropic ancilla, since otherwise, 
the initial and final state of such a process would not have the same entropy.

\aw{Let us thus consider a transformation
\begin{equation}
  \rho^{\otimes n} \longrightarrow \sigma^{\otimes m}\otimes \phi^{\otimes (n-m)},
\end{equation}}
where the number of copies $n$ and $m$ have to \aw{fulfil} the energy and entropy conservation
\aw{constraints 
\begin{equation}
\begin{split}
  E(\rho^{\otimes n}) &\sim E(\sigma^{\otimes m}\otimes \phi^{\otimes (n-m)}) , \\
  S(\rho^{\otimes n}) &\sim S(\sigma^{\otimes m}\otimes \phi^{\otimes (n-m)}) .
\end{split}
\end{equation}}
The above conditions can be easily written as a geometric equation
of the points $x_\psi=(E(\psi),S(\psi))$ with $\psi \in \{\rho,\sigma,\phi \}$
in the energy-entropy diagram:
\begin{equation}\label{eq:convex-combination}
  x_\rho = r \ x_\sigma + (1-r) \ x_\phi ,
\end{equation}
where $r\coloneqq m/n$ is the conversion rate, and we have only used
the extensivity of both entropy and energy in the number of copies, 
e.g.~$E(\rho^{\otimes n})= n E(\rho)$.

Eq.~\eqref{eq:convex-combination} implies that the three points $x_{\rho,\sigma,\phi}$ need to be aligned. 
In addition, the fact that $0\le r \le 1$ implies
that $x_\rho$ lies in the segment \aw{connecting} $x_\sigma$ and $x_\phi$ (see Fig.~\ref{fig:convex-combination}).
The conversion rate $r$ has then a geometric interpretation.
It is the relative Euclidean distance between $x_\phi$ and $x_\rho$ over the total distance $\|x_\sigma-x_\phi\|$, 
or in other words, \aw{it} is the fraction of the path that one has run when going \aw{from} 
$x_\phi$ to $x_\sigma$ and reaches $x_\rho$ (see Fig.~\ref{fig:convex-combination}).

\begin{figure}
\includegraphics[width=0.8\columnwidth]{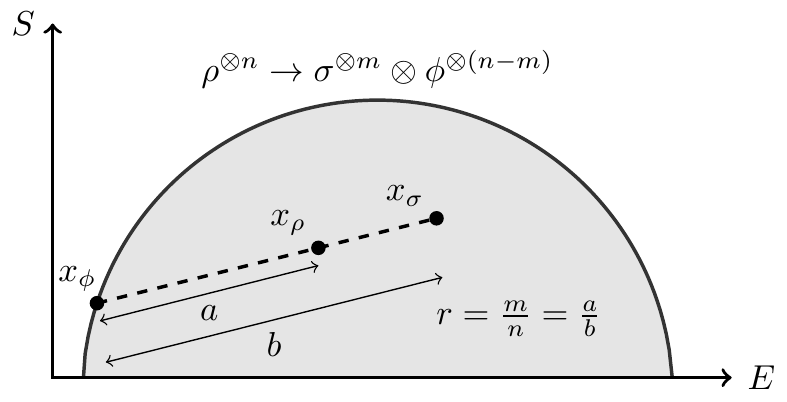}
\caption{Representation in the energy-entropy diagram of how the energy and entropy conservation
constraints in the transformation $\rho^{\otimes n}\to\sigma^{\otimes m}\otimes \phi^{\otimes n-m}$ 
imply $x_\rho$, $x_\sigma$ and $x_\phi$ to be aligned and $x_\rho$ to lie in between $x_\sigma$ and $x_\phi$.}
\label{fig:convex-combination}
\end{figure}

In order for the conversion rate $r$ from $\rho$ and $\sigma$ to be maximum, 
the state $\phi$ needs to lie in the boundary of the energy-entropy diagram, that is, 
it needs to be either a thermal or a pure state.
This allows us to quantitatively determine the conversion rate $r$ by
first finding the temperature of the thermal state $\phi$ which satisfies
\begin{equation}
  \frac{S(\sigma)-S(\phi)}{S(\sigma)-S(\rho)}=\frac{E(\sigma)-E(\phi)}{E(\sigma)-E(\rho)} ,
\end{equation}
and then
\begin{equation}\label{eq:interconvertibility-rate}
  \aw{r=\frac{S(\rho)-S(\phi)}{S(\sigma)-S(\phi)} .}
\end{equation}
See the geometrical interpretation in Fig.~\ref{fig:convex-combination}.

Note how Eq.~\eqref{eq:interconvertibility-rate} becomes \eqref{eq:rate-entropies} in the case
of $\phi$ being pure.
Let us also mention that, although at the first sight the inter-convertibility 
rate \eqref{eq:interconvertibility-rate} looks different from the one obtained in \cite{Sparaciari16}, 
it can be proven that they are the same. Eq.~\eqref{eq:interconvertibility-rate} is more compact 
than the one given \aw{there,} and its derivation much less technical.

\section{Thermodynamics with multiple conserved quantities}
\label{sec:charges}
The formalism for thermodynamics developed above can \aw{easily be} extended to situations
with multiple conserved quantities \aw{besides the energy, which we call ``charges''. These
could be particle numbers of certain chemical species, or physically conserved quantities
such as angular momentum. Such extensions to quantum thermodynamics, especially
in the resource setting, have been considered previously by several authors
\cite{Thermo-resource-charges-1,Thermo-resource-charges-2,Lostaglio2017,Guryanova2016,Halpern2016}.}

\aw{The setting is that each system comes with $q$ conserved quantities
$\hat L_k$ for $k=0,\ldots, q-1$, with $\hat Q_0=H$ the Hamiltonian. Here
we will assume that these quantities are all pairwise commuting. As before, when
we compose systems, the quantities are simply added: for another Hilbert space
$\mathcal{H}'$ with charges $\hat{L}_k'$ ($k=0,\ldots, q-1$), the joint Hilbert
space $\mathcal{H}\otimes\mathcal{H}'$ carries the global (or total) charges
$\hat{L}_k\otimes\1 + \1\otimes\hat{L}_k'$.}

\subsection{The charges-entropy diagram}
The energy-entropy diagram \aw{from the previous discussion is} extended to include the 
additional conserved quantities. Let us introduce the \emph{charges-entropy} diagram as:

\begin{defi}[Charges-entropy diagram]
The \emph{charges-entropy diagram} is the subset of points in $\mathbb{R}^{q+1}$ produced by 
all the states $\rho\in \mathcal{B}(\mathcal{H})$ with coordinates
\begin{equation}
  x_\rho = (L_0(\rho),L_1(\rho), \ldots,L_{q-1}, S(\rho)) ,
\end{equation}
where $L_k(\rho) \coloneqq \tr \hat L_k \rho$ for $k=0,\ldots,q-1$.
\aw{Note that in this section, to avoid confusion between the observable
(the operator $\hat{L}_k$) and the coordinate in the charges-entropy diagram
(the number $L_k$), we mark the former always by a hat.}
\end{defi}

\begin{figure}
\begin{center}
\vspace{.2cm}
\includegraphics[width=0.6\columnwidth]{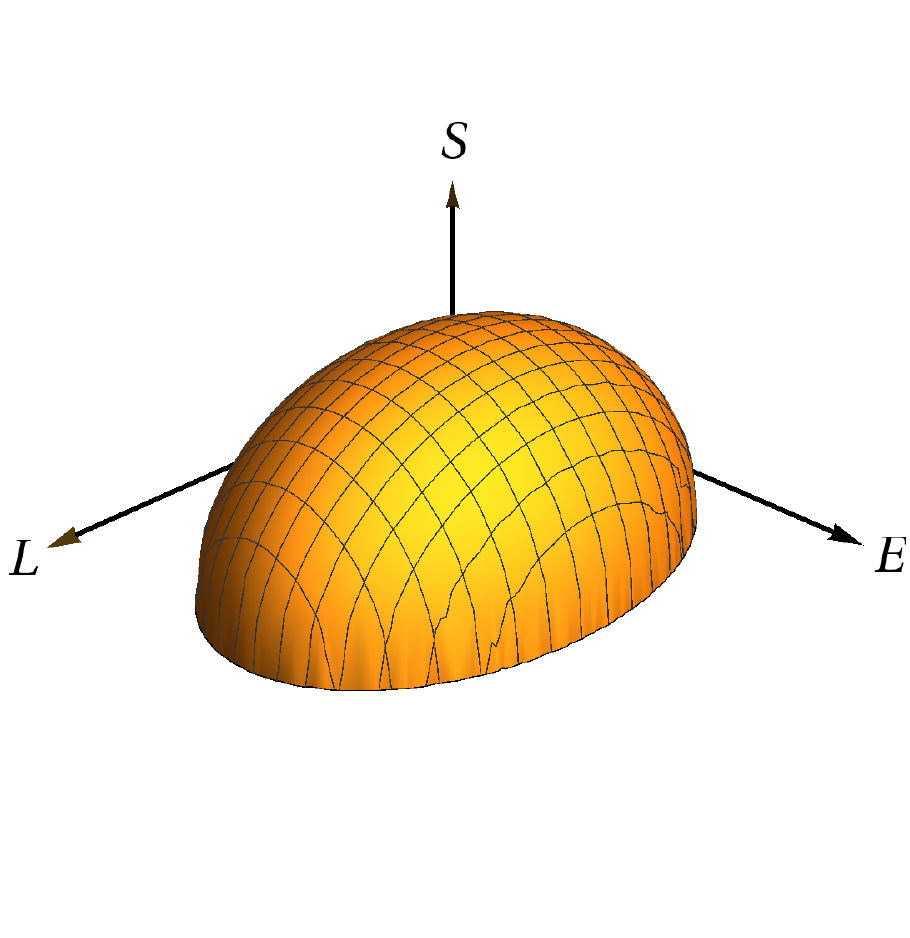}
\end{center}
\caption{\aw{Schematic of the charges-entropy} diagram of a system with two conserved quantities $E$ and $L$.}
\label{fig:thermo-diagram}
\end{figure}

\aw{Using the mutual commutation of the conserved quantities,}
we can easily show that the charges-entropy diagram forms a convex set in $\mathbb{R}^{q+1}$. 
To do so, let us first prove convexity for the zero entropy \aw{hyperplane},
which corresponds to the base of the diagram.
The zero entropy \aw{hypersurface} of the charges-entropy diagram is the set of points
in $\mathbb{R}^q$ with coordinates $x_\psi=(L_0(\psi),L_1(\psi),\ldots,L_{q-1}(\psi))$ for all 
\aw{normalized} pure states $\psi=\proj{\psi}$.
Let us consider in $\mathbb{R}^q$ the family of hyperplanes perpendicular to the unit vector 
$\vec \mu=(\mu_0,\ldots,\mu_{q-1})$ 
\begin{equation}\label{eq:hyper-plane}
 \vec \mu \cdot \vec L= \sum_{k=0}^{q-1} \mu_k L_k= C ,
\end{equation}
where $\vec L=(L_0,\ldots,L_{q-1})$ is a point in $\mathbb{R}^q$ and $C$ determines the \aw{hyperplane}. 
All points with coordinates $\vec L$ that fulfil Eq.~\eqref{eq:hyper-plane} 
belong to the hyperplane defined by $\vec \mu$ and $C$.
Given a direction $\vec \mu$,  there is an hyperplane that is particularly relevant, 
since it corresponds to the minimum possible value of $C$
\begin{equation}
C_{\min}(\vec \mu)=\min_{\psi\in {\cal H}}\,  
\frac{\bra{\psi}\vec \mu \cdot \vec{ \hat{L}}\ket{\psi}}{\langle\psi\ket \psi}
=\min_{\psi\in {\cal H}} \frac{\vec \mu \cdot \vec L(\psi)}{\langle\psi\ket \psi} .
\end{equation}
The minimum is reached by the eigenstate with smallest eigenvalue of the Hermitian operator 
$\vec \mu\cdot \hat L$ that we denote by $\ket{\psi (\vec \mu)}$, 
\begin{equation}
  \left( \sum_k \mu_k \hat L_k\right)\ket{\psi (\vec \mu)}= C_{\min}(\vec \mu) \ket{\psi (\vec \mu)} .
\end{equation}
\aw{Assume for the moment} that there is a unique ground state of $\vec \mu \cdot \vec{ \hat{L}}$. 
This means that the hyperplane defined by $C_{\min}(\vec \mu)$ is tangent to the charges-entropy diagram 
since it passes through one point of the diagram (corresponding to the single ground state 
$\ket{\psi(\vec \mu})$) leaving the rest of the diagram on the same side, 
i.e. $\vec \mu \cdot \vec L(\psi)\geqslant C_{\min}(\vec \mu)$ for all $\ket{\psi}\in {\cal H}$.

In case that there is a degenerate ground space described by a basis whose elements have different 
coordinates in the space of charges, the contact between
the hyperplane $C=C_{\min}$ and the zero entropy hypersurface is not a
point but a simplex of \aw{some higher} dimension, \aw{up to the dimension of the ground space
(i.e.~a segment in dimension $2$, triangle in dimension $3$, etc)}. 
As the conserved quantities are mutually commuting, there is a common eigenbasis
for all the $q$ charges.
\aw{Thus, there exists} a basis of the ground space of $\vec \mu\cdot \vec L$ which is also
eigenbasis of all the charges $L_k$ individually. If 
$\ket{\psi}$ and $\ket{\psi'}$ are two elements of that basis with different coordinates $\vec{L}$ and $\vec{L}'$, 
the state $\ket{\psi(\theta)}=\cos \theta \ket{\psi} + \sin\theta \ket{\psi'}$ with $\theta\in [0,\pi/2]$ has coordinates
\begin{equation}
  \label{eq:superposition_convex}
  \vec{L}(\psi(\theta))= (\cos^2\theta)\vec{L}  + (\sin^2\theta)\vec{L}',
\end{equation}
that is, the \aw{coordinates of $\psi(\theta)$} in the diagram are simply the corresponding 
convex combinations of the coordinates of $\vec L$ and $\vec L'$.

Property \eqref{eq:superposition_convex} holds for any two eigenstates of $\vec{\hat L}$,
in particular, for any two states $\ket{\psi}$ and $\ket{\psi'}$ that lie in the boundary 
of the zero-entropy hypersurface. 
Hence, as any point in the bulk of the zero-entropy hypersurface can be \aw{obtained} as a convex 
combination of two points of the boundary, all the points in the bulk belong to \aw{the} zero entropy surface.
Altogether, this proves that the zero-entropy surface is a convex set.

Once we understand the zero entropy hypersurface of the charges-entropy diagram, 
let us study its upper boundary.
The upper boundary of the charges-entropy diagram is described by 
the Generalized Gibbs Ensemble (GGE) instead of the canonical ensemble, i.e.
\begin{equation}
  \gamma(\vec \beta)\coloneqq \frac{1}{Z_{\vec \beta}}\e^{-\sum_k \beta_k \hat{L}_k}\, ,
\end{equation}
where $Z_{\vec \beta}\coloneqq \tr\bigl(\e^{-\sum_k \beta_k \hat{L}_k}\bigr)$ is the generalized partition function.
These states are precisely the states of maximum entropy among all \aw{those}
with prescribed expectation values $\langle \hat{L}_k\rangle$, $k=0,\ldots,q-1$.
We call this boundary the \emph{equilibrium boundary}.

The equilibrium boundary is mathematically described by all the points $\vec L(\rho)$
for which there exists a $\vec \beta \in \mathbb{R}^{q}$ such that
\begin{equation}
  \sum_{k=0}^{q-1} \beta_k L_k(\rho)-S(\rho) = -\log Z_\beta .
\end{equation}
Note that, given some $\vec \beta$, the points $\vec L(\rho)$ that fulfil
\begin{equation}\label{eq:family-hyperplanes}
  \sum_{k=0}^{q-1} \beta_k L_k(\rho)-S(\rho)=C ,
\end{equation}
belong to a hyperplane of dimension $q$. 
The hyperplane with $C=-\log Z_\beta$ is 
the one tangent to the equilibrium boundary.
The normal vector to this family of hyperplanes \eqref{eq:family-hyperplanes} is precisely $(\vec \beta, -1)$. 
This agrees with the fact that the tangent plane has a slope $\beta_k$ in the $k$-th direction, i.~e.
\begin{equation}\label{eq:partial-derivative}
  \frac{\partial S}{\partial L_k}=\beta_k \quad \forall\ k=0,\ldots,q-1 .
\end{equation}
Eq.~\eqref{eq:partial-derivative} can also be written in a vectorial form as
\begin{equation}
  \vec \beta = \vec \nabla S ,
\end{equation}
and $\vec \beta$ corresponds to the direction of maximal variation of the entropy.
This is shown in Fig.~\ref{fig:section-equilibrium-boundary}.

\begin{figure}
\begin{center}
\vspace{.2cm}
\includegraphics[width=0.8\columnwidth]{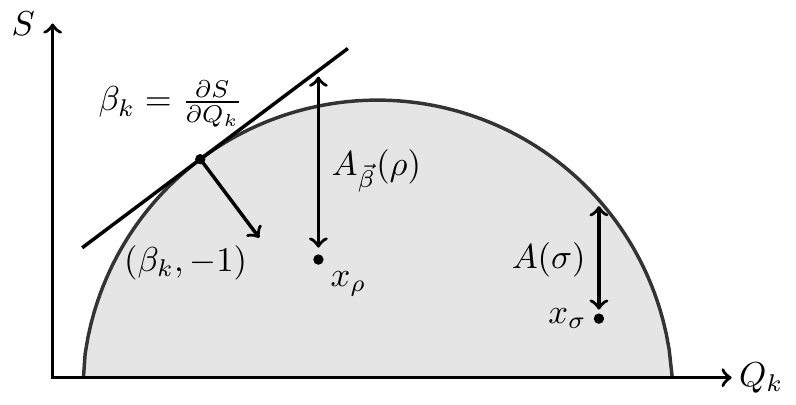}
\end{center}
\caption{Transverse section of the charges-entropy diagram in the \aw{$L_k$-$S$ plane}.
The normal vector to the tangent plane has coordinates $(\vec \beta, -1)$, 
which in the section appears as $(\beta_k, -1)$.
The $\vec \beta$-athermality, $A_{\vec \beta}(\rho)$, and the absolute athermality, 
$A(\sigma)$, are respectively represented for two states $\rho$ and $\sigma$.}
\label{fig:section-equilibrium-boundary}
\end{figure}

Note that in order to microscopically justify the charges-entropy diagram,
it must be shown that states with the same
expectation values for the conserved quantities
can be connected to each other using unitaries that commute with all charges,
and using sublinear ancillas, in the limit of many copies.
\aw{To state the theorem, we fix some notation. Consider a quantum system with mutually commuting 
charges $\hat{L}_0=H$, $\hat{L}_1,\ldots, \hat{L}_{q-1}$ on a finite dimensional Hilbert space $\mathcal{H}$.
The corresponding description of $n$ copies of this system has Hilbert space
$\mathcal{H}^{\otimes n}$ and charges
\begin{equation}
  \hat{L}_k^{(n)} = \sum_{i=1}^n \1_{\mathcal{H}^{\otimes i-1}} \otimes (\hat{L}_k)_i \otimes \1_{\mathcal{H}^{\otimes n-i}}.
\end{equation}}

\begin{theorem} 
\label{main theorem}
Two states $\rho$ and $\sigma$ of the above quantum system have the same 
entropy $S(\rho)=S(\sigma)$ and the average charge values $\Tr(\hat{Q}_k \rho)=\Tr(\hat{Q}_k \sigma)$ 
for all $k=0,\ldots,q-1$ if and only if for all $n$ there exists an 
ancillary quantum system with the Hilbert space $\mathcal{K}$ of dimension $2^{o(n)}$,
with charges $\hat{L}_0'=H$', $\hat{L}_1',\ldots, \hat{L}_{q-1}'$,
and a charge conserving unitary $U$ acting on $\mathcal{H}^{\otimes n} \otimes \mathcal{K}$,
such that 
\begin{align}
  \lim_{n\to \infty} \norm{U (\rho^{\otimes n} \otimes \omega') U^{\dagger} - \sigma^{\otimes n} \otimes \omega}_1 &= 0, \\
  \norm{\hat{L}_k'}_{\infty} \leq o(n) \ \text{ and }\ [U,\hat{L}_k^{(n)}+\hat{L}_k'] &=0 \quad \forall k,
\end{align}
where $\omega$ and $\omega'$ are suitable states of the ancillary system. 
\end{theorem}

\aw{The proof of this theorem (which generalises the asymptotic equivalence theorem of
\cite{Sparaciari16}) is presented in our forthcoming work \cite{MultipleChargesThms}.
There, we extensively study this issue, and show how to generalize it even further
from commuting to non-commuting conserved charges. In the following, we limit ourselves 
to the main points.}

\subsection{Athermality and free entropy}
Proceeding as \aw{before}, let us introduce some quantities that will be relevant
in the following.

\begin{defi}[$\vec \beta$-athermality]
The $\vec \beta$-athermality of a state $\rho$ is defined as
\begin{equation}
  A_{\vec \beta}(\rho) \coloneqq \sum_{k=0}^{q-1} \beta_k L_k(\rho)-S(\rho)+\log Z_{\vec \beta} .
\end{equation}
\end{defi}

The $\vec \beta$-athermality of a point with coordinates $(\vec L(\rho),S(\rho))$ 
can be interpreted geometrically in the charges-entropy diagram 
as the vertical distance from the hyperplane tangent to the equilibrium boundary 
with normal vector $(\vec \beta, -1)$. This is represented in Fig.~\ref{fig:section-equilibrium-boundary}.
Note that the $\vec \beta$-athermality is also introduced in \cite{Guryanova2016}
\aw{under the name} \emph{free entropy}.

The $\vec \beta$-athermality can also be written as
\begin{equation}
  A_{\vec \beta}(\rho)=D(\rho\|\gamma(\vec \beta)) ,
\end{equation}
with $D(\rho\|\sigma)\coloneqq \tr(\rho \log \rho-\rho \log \sigma)$ being the relative entropy.
Hence, because of the \aw{non-negativity} of the relative entropy $D(\rho\|\sigma)\geqslant 0$, 
we have that
\begin{equation}
  S(\rho) \leqslant \sum_{k=0}^{q-1} \beta_k L_k(\rho)+\log \aw{Z_{\vec\beta}}.
\end{equation}
Note now that the right hand side corresponds to the entropy coordinate of the hyperplane
with normal vector $\vec \beta$ tangent to the equilibrium boundary.
Thus, any tangent hyperplane with normal vector $\vec \beta$ leaves all the points of the 
charges-entropy diagram below it. This \aw{is consistent with what we saw before,} 
that the charges entropy diagram is a convex set.

In a similar way we define the \emph{absolute athermality} or simply athermality 
in the following.

\begin{defi}[Absolute athermality]
The absolute athermality of a state $\rho$ is defined as
\begin{equation}
  A(\rho)\coloneqq \min_{\vec \beta \in \mathbb{R}^q} A_{\vec \beta}(\rho) .
\end{equation}
\end{defi}

The athermality of a state $\rho$ can be understood geometrically as the vertical distance 
from the thermal boundary (see Fig.~\ref{fig:section-equilibrium-boundary}).

\subsection{Charge extraction}
A first relevant scenario of thermodynamics with multiple conserved quantities is the extraction 
of a charge of the system while keeping the \aw{other charges constant}.
For \aw{this subset} of operations, the system is restricted to move along a straight line 
in the direction of the extracted charge within the \aw{iso-entropic} hyperplane (see Fig.~\ref{fig:constant-entropy-section}).

The bound charge for such \aw{a} scenario is given by
\begin{equation}\label{eq:bound-charge-single}
 B_k(\rho) \coloneqq  \min_{\substack{\sigma \, : \, S(\sigma)=S(\rho), \\ L_i(\sigma)=L_i(\rho)\ \forall i\ne k}} L_k(\sigma),
\end{equation}
where $\gamma_k(\rho)$ is the GGE state that attains the minimum and corresponds to the
point of the equilibrium boundary which is intersected by the straight-line with direction $k$
which passes through $\rho$. \aw{As in the case of only energy, it is easy to see that
$B_k(\rho)=L_k(\gamma_k(\rho))$.}
This is represented in Fig.~\ref{fig:constant-entropy-section} (left).
The direction $\vec \beta$ corresponding to $\gamma_k$ can be geometrically determined 
as the normal vector of the tangent plane to the equilibrium surface in that point. 

The free charge $F_k(\rho)$ is \aw{now analogously defined as}
\begin{equation}
  F_k(\rho)= L_k(\rho)- B_k(\rho),
\end{equation}
and corresponds to the maximum amount of charge that can be extracted given the
restrictions of constant entropy and charges.
This is represented in Fig.~\ref{fig:constant-entropy-section} (left) for a case of 2 
conserved quantities.

A detailed and full rigorous analysis of 
the charge extraction by means of unitary conserving operations in the 
many copy limit (for both commuting and non-commuting charges) will be made in our 
forthcoming work \aw{\cite{MultipleChargesThms}}.

\begin{figure}
\begin{center}
\vspace{.2cm}
\includegraphics[width=0.45\columnwidth]{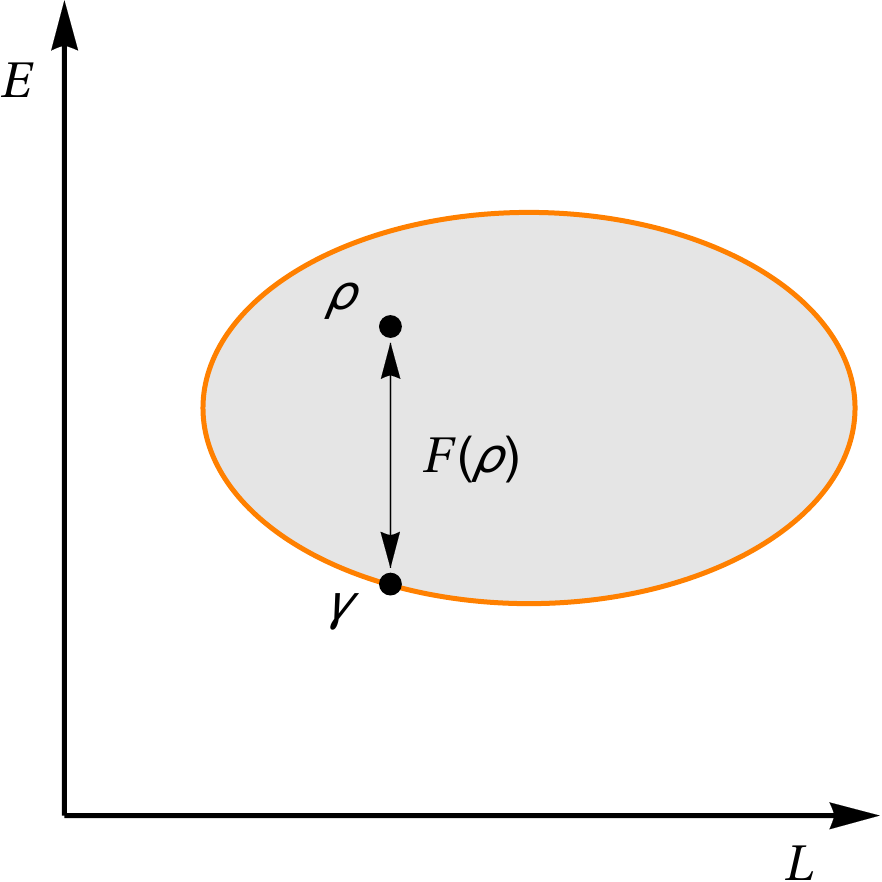}
\includegraphics[width=0.45\columnwidth]{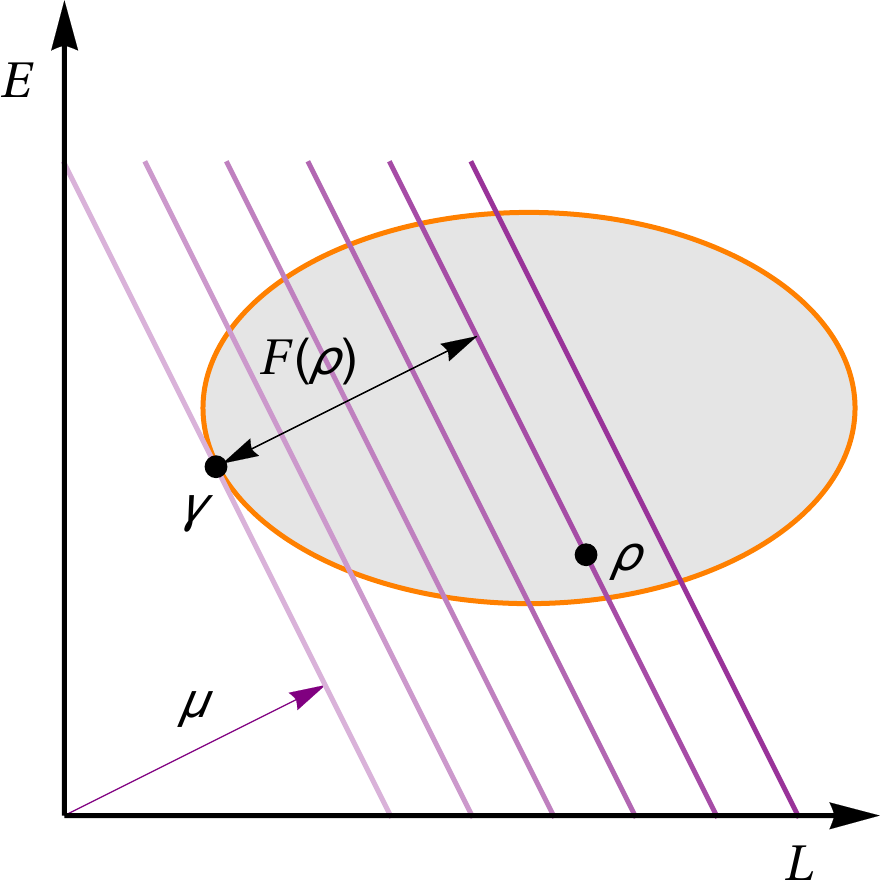}
\end{center}
\caption{Constant entropy section of a charges-entropy diagram with 2 charges $E$ and $L$.
(Left) Scenario of extraction of a single charge while keeping constant the
other charges. The system is constrained to move in one dimensional line and free-charge $F(\rho)$ of a
state $\rho$ is the distance from the boundary in such direction.
(Right) Scenario in which the charges are allowed to change and the aim is the extraction 
of a potential $V_{\hat \mu}(\rho)=\sum_{k=0}^{q-1} \mu_k L_k(\rho)$. 
The parallel purple lines represent equipotential surfaces of $V_{\hat \mu}$. 
The GGE state $\gamma$ is the state which minimizes $V$ and belongs to the hyperplane which is
tan tangent to the equilibrium boundary.}
\label{fig:constant-entropy-section}
\end{figure}

\subsection{Extraction of a \aw{generalized potential}}
An different, but also relevant scenario is when the set of operations
\aw{allowed for the variation of the charges is unrestricted, and the aim
is the extraction of a \emph{generalized potential}; cf.~\cite{Guryanova2016,Halpern2016}}
\begin{equation}
  V_{\vec\mu}(\rho)\coloneqq\sum_{k=0}^{q-1} \mu_k L_k(\rho),
\end{equation}
where $\vec\mu\coloneqq (\mu_0,\ldots,\mu_{q-1})$ specifies the weight $\mu_k$ of every charge 
$L_k$ in the linear combination $V_{\vec\mu}$.
For convenience, $\vec\mu$ is assumed to be normalized $\norm{\vec\mu}=1$ in the Euclidean norm
and its components to be positive $\mu_k\geqslant 0$.
(An alternative normalization could be to take the coefficient of the Hamiltonian $\mu_0=1$ as
it is usually \aw{done in the grand canonical} ensemble.)

In this setting, the \emph{bound potential} is defined as
\begin{equation}\label{eq:bound-charge}
 B_{\vec\mu}(\rho) \coloneqq  \min_{\sigma \, : \, S(\sigma)=S(\rho)} V_{\vec\mu}
 (\sigma) ,
\end{equation}
where $\gamma_{\vec\mu}(\rho)$ is the state that attains the minimum which is again a GGE.
\aw{As before, $B_{\vec\mu}(\rho) = V_{\vec\mu}(\gamma_{\vec\mu}(\rho))$.}
Note that the $\vec \beta$ values of the state $\gamma_{\vec\mu}(\rho)$ need to be proportional
to the \aw{unit} vector $\vec\mu$
\begin{equation}
  \vec \beta = \beta \vec\mu ,
\end{equation}
with $\beta$ a scalar that is determined by the \aw{condition of equal entropies, $S(\rho)=S(\gamma_{\vec\mu}(\rho))$}.
This can be seen geometrically in Fig.~\ref{fig:constant-entropy-section} (right)
or analytically by making the simple observation that, 
\aw{regarding $\tilde H \coloneqq V_{\vec\mu}$ as a new Hamiltonian},
the \aw{min-energy} principle singles out $\e^{-\beta \tilde H}$ as the
state that attains the minimum.

Analogously, the free potential is given by
\begin{equation}
F_{\vec\mu}(\rho)= V_{\vec\mu}(\rho)- B_{\vec\mu}(\rho),
\end{equation}
and corresponds to the maximum amount of generalized potential $V_{\vec\mu}$ 
that can be extracted under entropy preserving operations.
This situation has been diagrammatically represented in Fig.~\ref{fig:constant-entropy-section}.

The scenario with a generalized potential $V_{\vec\mu}$ is analogue to the single charge situation. 
Both first and second laws can be stated as in Secs.~\ref{sec:first} and \ref{sec:second},
replacing the free energy $F(\rho)$ by $F_{\vec\mu}(\rho)$.

\subsection{Second law}
Any process (not necessarily entropy preserving) 
that brings an initial generalized Gibbs state $\gamma_B(\vec \beta)$ out of equilibrium fulfils
\begin{equation}\label{eq:2nd-law-multiple-charges}
  \sum_k\beta_k \Delta L_k^B \geqslant \Delta S_B ,
\end{equation}
where the inequality is only saturated in the limit of large baths (small variations of entropy) 
and when final state is also at equilibrium (GGE).
Geometrically, Eq.~\eqref{eq:2nd-law-multiple-charges} is a trivial consequence that the charges entropy
diagram is upper bounded by any of its tangent planes.

In the particular case of an entropy preserving process 
on an initial bipartite state $\rho_A\otimes \gamma_B(\vec \beta)$, 
the change in mutual information between the subsystems $A$ and $B$ \aw{is}
\begin{equation}
\Delta I = \Delta S_A+\Delta S_B \geqslant 0 ,
\end{equation}
which together with Eq.~\eqref{eq:2nd-law-multiple-charges} implies
\begin{equation}\label{eq:2nd-law-system}
  \sum_k\beta_k \Delta L_k^B \geqslant -\Delta S_A .
\end{equation}
Note that the above equation \eqref{eq:2nd-law-system} is equivalent to
the second law formulated in \cite{Guryanova2016}.
We see here that such a law is also valid for baths of arbitrary small sizes.
Eq.~\eqref{eq:2nd-law-system} is \aw{asymptotically saturated with equality} in the limit of large system sizes and
in \aw{the} absence of correlations in the final state $\rho_{AB}'=\rho_A'\otimes \gamma_B'$.

\subsection{\aw{Interconvertibility} rates}
The \aw{rate of interconversion} from a quantum state $\rho$ to a state $\sigma$, \aw{in the
same sense as in Section \ref{sec:resource}, i.e.
\begin{equation}
  \rho^{\otimes n} \longrightarrow \sigma^{\otimes m}\otimes \phi^{\otimes (n-m)},
\end{equation}
is given by the conservation of the entropy and the charges $L_k$:
\begin{equation}
\begin{split}
  L_k(\rho^{\otimes n}) &\sim L_k(\sigma^{\otimes m}\otimes \phi^{\otimes (n-m)}) \ \forall \, k=0,\ldots,q-1 , \\
  S(\rho^{\otimes n})   &\sim S(\sigma^{\otimes m}\otimes \phi^{\otimes (n-m)}) .
\end{split}
\end{equation}}
The above conditions can again be written as a geometric equation
\aw{for the points $x_\xi=(L_0(\xi),\ldots,L_{q-1}(\xi),S(\xi)) \in \mathbb{R}^{q+1}$ in the energy-entropy diagram, 
with $\xi \in \{\rho,\sigma,\phi \}$:}
\begin{equation}\label{eq:convex-combination-multiple}
  x_\rho = r \ x_\sigma + (1-r) \ x_\phi ,
\end{equation}
where we have merely used the extensivity of both entropy and the conserved quantities in the 
number of copies, e.g.~$S(\rho^{\otimes n})= n S(\rho)$.

Thus, by means of the same argument used in the previous section, 
the convertibility rate from $\rho$ to $\sigma$ reads
\begin{equation}\label{eq:interconvertibility-rate-multiple}
  r = \lim_{n\rightarrow\infty} \frac{m}{n} = \frac{S(\rho)-S(\phi)}{S(\sigma)-S(\phi)},
\end{equation}
where the state $\phi$ is determined from the following set of $q$ equations
\begin{equation}
  \frac{S(\sigma)-S(\phi)}{S(\sigma)-S(\rho)} = \frac{L_k(\sigma)-L_k(\phi)}{L_k(\sigma)-L_k(\rho)} 
  \quad \forall\, k=0,\ldots,q-1 .
\end{equation}
The state $\phi$ can be found geometrically in the charges-entropy diagram as the point in the 
boundary \aw{of the charges-entropy diagram} that 
\aw{intersects the straight half-line that goes through $\rho$ from $\sigma$
(see Fig.~\ref{fig:convex-combination} for the single charge example).}

\section{Discussion}
\label{sec:discuss}
In the recent years, thermodynamics, and in particular work extraction from non-equilibrium states, 
has been studied in the quantum domain, giving rise to radically new insights into quantum thermal 
processes. However, in \aw{the standard treatment} of thermodynamics, be it classical or quantum, 
thermal baths are considered to be \aw{asymptotically} large in size. That is why if a system is 
attached \aw{to} a bath and allowed to exchange energy and entropy, the bath stays intact and its 
temperature remains unchanged. 
That is also why the equilibrated system finally acquires the same temperature \aw{as} the bath. 
However, if one goes beyond this assumption and considers \aw{the} bath to be a finite system, \aw{too,} 
then traditional thermodynamics breaks down. This situation is very much relevant for thermodynamics 
in the quantum regime, where both system and bath may be small. 
The first problem that appears in such a situation is the notion of temperature itself, since the 
finite bath may go out of thermal equilibrium due to the exchange of energy with the system. 
Therefore, it is absolutely necessary to develop a temperature independent universal thermodynamics, 
in which the bath could be small or large,  and would not  get any special status.   

In this work, we have formulated temperature independent thermodynamics as an exclusive consequence 
of information conservation. We have relied on the fact that, for a given amount of information, 
measured by the von Neumann entropy, any system can only be transformed to states with the same entropy. 
Given this constraint of information conservation, there is a \aw{distinguished} state within the constant 
entropy manifold, which is the state that possesses minimal energy. 
This state is known as a {\it completely passive state} and acquires \aw{the} \aw{Boltzmann-Gibbs} canonical 
form with an intrinsic temperature. 
We call the energy of the completely passive state the bound energy, since no further energy can be 
extracted by means of entropy preserving operations.
Thus, for a given state, the difference between its energy and bound energy corresponds to the maximum 
amount of energy that can be extracted in form of work and we \aw{call it} free energy.
In fact, in this framework, two states that have the same entropy and energy are thermodynamically 
equivalent \cite{Sparaciari16}.
The thermodynamic equivalence between equal entropy and energy states has allowed us to use of the 
energy-entropy diagram to illustrate the notions of bound and free energies in a geometric way.
We have introduced a new definition of heat for arbitrary systems in terms of \aw{the bound energy
of the bath}.

We have seen that the laws of thermodynamics are a consequence of the reversible dynamics of the 
underlying physical theory. In particular:
\begin{itemize}
\item \emph{Zeroth law} emerges as the consequence of information conservation.
\item \emph{First and second \aw{law}} emerge as the consequence of energy conservation, together with information conservation.
\item \emph{Third law} emerges as the consequence of \aw{fine-grained} information conservation (microscopic reversibility or unitarity). 
      Therefore, there is no third law if one considers \aw{coarse-grained} information conservation, \aw{as we did here}. 
\end{itemize} 

We have demonstrated that the maximum efficiency of a quantum engine with a finite bath is in general
lower than that of an ideal Carnot engine. We have introduced a resource theoretic framework for our 
intrinsic thermodynamics, within which we address the problem of work extraction and inter-state transformations. 
All these results have been illustrated by means of the energy-entropy diagram. 
Furthermore, we \aw{gave} a geometric interpretation in the diagram \aw{of} 
the relevant thermodynamic quantities,
as well as the inter-convertibility rate between quantum states under entropy and energy preserving operations.

The information conservation based framework for thermodynamics, as well as the resource theory and the 
energy-entropy diagram, is also extended to multiple conserved quantities, \aw{as long as they commute with each other}. 
In this case, the energy-entropy diagram becomes the charges-entropy diagram and allows us to understand 
thermodynamics in a geometrical way.
In particular, we have studied the extraction of a single charge while keeping the other charges 
conserved as well as the extraction of a generalized potential. 
In the first scenario, we have seen that the maximum work
extractable from any state by operations that asymptotically
conserve the given charges is the difference between the
free energy of the state and that of the \aw{iso-entropic} generalized
grand canonical Gibbs state.
Concerning the extraction of a generalized potential (a linear combination of charges),
we have shown that it is analogous to the work extraction (the single charge case).
Finally, we have determined the \aw{interconvertibility} rates between states
with different entropy and charges. \aw{A deeper investigation of the setting
with multiple charges, including the generalisation to non-commuting conserved
quantities, is left to our forthcoming paper \cite{MultipleChargesThms}.}

In general, thermodynamics can be studied in three different scenarios:
\begin{itemize}
\item \emph{One-shot or single-copy \aw{setting}}, where only one copy of \aw{the} joint system-environment is available. 
  In this case, even the notion of expectation value is not meaningful, as well as von Neumann entropy.

\item \emph{Limit of \aw{many runs}}, where there are many copies, but operation are restricted to \aw{single copies of the system}. 
  In this scenario, the notions of expectation value and von Neumann entropy are well defined.  

\item \emph{Limit of \aw{many copies}}, where one has access to arbitrarily many copies of the system and an
  ancilla sub-linear in the number of copies that can globally be processed. 
\end{itemize}

A first observation is that our formalism cannot be applied in the single-shot \aw{setting}, 
since the notion of expectation value (say energy) cannot in general be used.

A relevant point to discuss is what happens when operations are not \EP but unitaries.
In the limit of many copies, \aw{state transformations under unitaries} converge to \EP operations  
and our formalism is recovered.
The limit of many runs is a bit more subtle. 
On \aw{the} one hand, all the thermodynamic inequalities of our formalism are respected, 
since unitaries form a subset of \EP operations.
On the other hand, these thermodynamic inequalities will not be in general saturated.
For instance, our formalism states that the work that can be extracted per system in a 
state $\rho$ is upper bounded by its free energy $W \leqslant F(\rho)$. 
In the \aw{single-copy or many-run settings} with fine-grained information conservation, the law will be respected, 
but there will not be in general a unitary for which $W=F(\rho)$.

A natural open question is to what extent our formalism can be extended 
from considering coarse-grained information conservation operations as the set of allowed operations to unitaries.
In that case, the notion of bound energy would be different and
many more equivalence classes of states would appear.
Something similar already happens in the resource theory of thermodynamics, where instead of having a single monotone as an
``if and only if'' condition for state transformation, infinitely many are required \cite{Brandao15}.
It is far from clear whether under fine-grained information conservation restriction
the energy-entropy diagrams (or a generalization of them) would still be useful.

Let us finally point out that, in the extension of our work to unitary operations in the 
settings of single-shot and many-runs, a consistent formulation of \aw{the} zeroth law would not be possible.
Recall that \aw{the} zeroth law states that a collection of systems are in mutual thermal equilibrium 
if and only if their arbitrary combinations are also in equilibrium. 
It is well known that passive states that are not thermal do not remain passive when sufficiently 
many copies are considered. Hence, for establishing a consistent zeroth law, one has to consider
operations beyond unitaries on a single copy.

\section*{Acknowledgements}
We thank  Ph. Faist, K. \aw{Gaw\c{e}dzki}, R. B. T. Harvey, J. Kimble, S. Maniscalco, 
Ll. Masanes, J. Oppenheim, V. Pellegrini, M. Polini, C. Sparaciari, A. Vulpiani, \aw{N. Yunger Halpern} and R. Zambrini
for useful discussions and comments in both theoretical and experimental aspects of our work. 
We also thank the referees for constructive comments.

We acknowledge financial support from the European Commission 
(FETPRO QUIC H2020-FETPROACT-2014 No. 641122), 
the European Research Council (AdG OSYRIS and AdG IRQUAT), 
the Spanish MINECO (grants no. FIS2008-01236, FISICATEAMO FIS2016-79508-P,
FIS2013-40627-P, FIS2016-86681-P,
and Severo Ochoa Excellence Grant SEV-2015-0522) with the support of FEDER funds, 
the Generalitat de Catalunya (grants no.~2017 SGR 1341, and SGR 875 and 966), CERCA Program/Generalitat  de  Catalunya
and Fundaci{\'o} Privada Cellex.
AR thanks for support from the CELLEX-ICFO-MPQ fellowship.

\bibliographystyle{unsrtnat}

\begin{thebibliography}{56}
\providecommand{\natexlab}[1]{#1}
\providecommand{\url}[1]{\texttt{#1}}
\expandafter\ifx\csname urlstyle\endcsname\relax
  \providecommand{\doi}[1]{doi: #1}\else
  \providecommand{\doi}{doi: \begingroup \urlstyle{rm}\Url}\fi

\bibitem[Gemmer et~al.(2009)Gemmer, Michel, and Mahler]{Gemmer09}
Jochen Gemmer, M.~Michel, and G\"unther Mahler.
\newblock \emph{Quantum Thermodynamics}, volume 748 of \emph{Lecture Notes in
  Physics}.
\newblock Berlin, Heidelberg: Springer, 2009.
\newblock ISBN 9783540705093.
\newblock \doi{10.1007/978-3-540-70510-9}.
\newblock URL \url{http://link.springer.com/book/10.1007%2F978-3-540-70510-9}.

\bibitem[Flanders and Swann(1964)]{FS:heat-and-work}
Michael Flanders and Donald Swann.
\newblock {The First and Second Law of Thermodynamics}.
\newblock In \emph{At the Drop of Another Hat}. Parlophone Ltd., 1964.
\newblock URL \url{https://youtu.be/VnbiVw_1FNs}.

\bibitem[Parrondo et~al.(2015)Parrondo, Horowitz, and Sagawa]{Parrondo15}
Juan M.~R. Parrondo, Jordan~M. Horowitz, and Takahiro Sagawa.
\newblock Thermodynamics of information.
\newblock \emph{Nature Physics}, 11:\penalty0 131--139, 2015.
\newblock \doi{doi:10.1038/nphys3230}.
\newblock URL \url{http://dx.doi.org/10.1038/nphys3230}.

\bibitem[Maxwell(1908)]{Maxwell08}
James~Clerk Maxwell.
\newblock \emph{Theory of Heat}.
\newblock Longmans, Green, and Co.: London, New York, Bombay, 1908.

\bibitem[Leff and Rex(1990)]{Leff90}
Harvey~S. Leff and Andrew~F. Rex.
\newblock \emph{Maxwell's Demon: Entropy, Information, Computing}.
\newblock Princeton University Press, 1990.
\newblock ISBN 9780691605463.
\newblock URL \url{http://press.princeton.edu/titles/4731.html}.

\bibitem[Leff and Rex(2002)]{Leff02}
Harvey~S. Leff and Andrew~F. Rex.
\newblock \emph{Maxwell's Demon 2: Entropy, Classical and Quantum Information,
  Computing}.
\newblock Taylor and Francis, London, 2002.
\newblock ISBN 9780750307598.
\newblock URL \url{https://www.taylorfrancis.com/books/9780750307598}.

\bibitem[Maruyama et~al.(2009)Maruyama, Nori, and Vedral]{Maruyama09}
Koji Maruyama, Franco Nori, and Vlatko Vedral.
\newblock \textit{Colloquium}: {The physics of Maxwell's demon and
  information}.
\newblock \emph{Reviews in Modern Physics}, 81:\penalty0 1--23, Jan 2009.
\newblock \doi{10.1103/RevModPhys.81.1}.
\newblock URL \url{http://link.aps.org/doi/10.1103/RevModPhys.81.1}.

\bibitem[Szilard(1929)]{Szilard29}
Leonard Szilard.
\newblock {\"U}ber die {E}ntropieverminderung in einem thermodynamischen
  {S}ystem bei {E}ingriffen intelligenter {W}esen.
\newblock \emph{Zeitschrift f\"ur Physik}, 53\penalty0 (11):\penalty0 840--856,
  1929.
\newblock \doi{10.1007/BF01341281}.
\newblock URL \url{http://dx.doi.org/10.1007/BF01341281}.

\bibitem[Landauer(1961)]{Landauer61}
Rolf Landauer.
\newblock {Irreversibility and Heat Generation in the Computing Process}.
\newblock \emph{IBM Journal of Research and Development}, 5\penalty0
  (3):\penalty0 183--191, July 1961.
\newblock ISSN 0018-8646.
\newblock \doi{10.1147/rd.53.0183}.

\bibitem[Bennett(1982)]{Bennett82}
Charles~H. Bennett.
\newblock The thermodynamics of computation---a review.
\newblock \emph{International Journal of Theoretical Physics}, 21\penalty0
  (12):\penalty0 905--940, 1982.
\newblock ISSN 1572-9575.
\newblock \doi{10.1007/BF02084158}.
\newblock URL \url{http://dx.doi.org/10.1007/BF02084158}.

\bibitem[Plenio and Vitelli(2001)]{Plenio01}
Martin~B. Plenio and V.~Vitelli.
\newblock The physics of forgetting: Landauer's erasure principle and
  information theory.
\newblock \emph{Contemporary Physics}, 42\penalty0 (1):\penalty0 25--60, 2001.
\newblock \doi{10.1080/00107510010018916}.
\newblock URL \url{http://dx.doi.org/10.1080/00107510010018916}.

\bibitem[del Rio et~al.(2011)del Rio, {\AA}berg, Renner, Dahlsten, and
  Vedral]{Rio11}
Lidia del Rio, Johan {\AA}berg, Renato Renner, Oscar C.~O. Dahlsten, and Vlatko
  Vedral.
\newblock The thermodynamic meaning of negative entropy.
\newblock \emph{Nature}, 474:\penalty0 61--63, 2011.
\newblock \doi{10.1038/nature10123}.
\newblock URL \url{http://dx.doi.org/10.1038/nature10123}.

\bibitem[Reeb and Wolf(2014)]{Reeb14}
David Reeb and Michael~M. Wolf.
\newblock An improved {L}andauer principle with finite-size corrections.
\newblock \emph{New Journal of Physics}, 16\penalty0 (10):\penalty0 103011,
  2014.
\newblock \doi{10.1088/1367-2630/16/10/103011}.
\newblock URL
  \url{http://iopscience.iop.org/article/10.1088/1367-2630/16/10/103011}.

\bibitem[Shannon(1948)]{Shannon48}
Claude~E. Shannon.
\newblock A mathematical theory of communication.
\newblock \emph{Bell System Technical Journal}, 27\penalty0 (3):\penalty0
  379--423, July 1948.
\newblock ISSN 0005-8580.
\newblock \doi{10.1002/j.1538-7305.1948.tb01338.x}.

\bibitem[Nielsen and Chuang(2000)]{Nielsen00}
Michael~A. Nielsen and Isaac~L. Chuang.
\newblock \emph{Quantum Computation and Quantum Information}.
\newblock Cambridge: Cambridge University Press, 2000.
\newblock ISBN 978-0521635035.

\bibitem[Cover and Thomas(2005)]{Cover05}
Thomas~M. Cover and Joy~A. Thomas.
\newblock \emph{Elements of Information Theory}.
\newblock John Wiley and Sons, Inc., 2 edition, 2005.
\newblock ISBN 9780471748823.
\newblock \doi{10.1002/047174882X}.
\newblock URL \url{http://dx.doi.org/10.1002/047174882X}.

\bibitem[Alicki and Fannes(2013)]{Alicki13}
Robert Alicki and Mark Fannes.
\newblock Entanglement boost for extractable work from ensembles of quantum
  batteries.
\newblock \emph{Physical Review E}, 87:\penalty0 042123, Apr 2013.
\newblock \doi{10.1103/PhysRevE.87.042123}.
\newblock URL \url{http://link.aps.org/doi/10.1103/PhysRevE.87.042123}.

\bibitem[Perarnau-Llobet et~al.(2015)Perarnau-Llobet, Hovhannisyan, Huber,
  Skrzypczyk, Brunner, and Ac\'{\i}n]{Marti15}
Mart\'{\i} Perarnau-Llobet, Karen~V. Hovhannisyan, Marcus Huber, Paul
  Skrzypczyk, Nicolas Brunner, and Antonio Ac\'{\i}n.
\newblock {Extractable Work from Correlations}.
\newblock \emph{Physical Review X}, 5:\penalty0 041011, Oct 2015.
\newblock \doi{10.1103/PhysRevX.5.041011}.
\newblock URL \url{http://link.aps.org/doi/10.1103/PhysRevX.5.041011}.

\bibitem[Bera et~al.(2017{\natexlab{a}})Bera, Riera, Lewenstein, and
  Winter]{Bera16}
Manabendra~Nath Bera, Arnau Riera, Maciej Lewenstein, and Andreas Winter.
\newblock Generalized laws of thermodynamics in the presence of correlations.
\newblock \emph{Nature Communications}, 8:\penalty0 2180, 2017{\natexlab{a}}.
\newblock \doi{doi:10.1038/s41467-017-02370-x}.
\newblock URL \url{https://www.nature.com/articles/s41467-017-02370-x}.

\bibitem[Short(2011)]{Short11}
Anthony~J. Short.
\newblock Equilibration of quantum systems and subsystems.
\newblock \emph{New Journal of Physics}, 13\penalty0 (5):\penalty0 053009,
  2011.
\newblock \doi{10.1088/1367-2630/13/5/053009}.
\newblock URL
  \url{http://iopscience.iop.org/article/10.1088/1367-2630/13/5/053009}.

\bibitem[Goold et~al.(2016)Goold, Huber, Riera, del Rio, and
  Skrzypczyk]{Goold16}
John Goold, Marcus Huber, Arnau Riera, Lidia del Rio, and Paul Skrzypczyk.
\newblock The role of quantum information in thermodynamics: a topical review.
\newblock \emph{Journal of Physics A: Mathematical and Theoretical},
  49\penalty0 (14):\penalty0 143001, 2016.
\newblock \doi{10.1088/1751-8113/49/14/143001}.
\newblock URL
  \url{http://iopscience.iop.org/article/10.1088/1751-8113/49/14/143001}.

\bibitem[del Rio et~al.(2016)del Rio, Hutter, Renner, and Wehner]{Rio16}
L\'{\i}dia del Rio, Adrian Hutter, Renato Renner, and Stephanie Wehner.
\newblock Relative thermalization.
\newblock \emph{Physical Review E}, 94:\penalty0 022104, Aug 2016.
\newblock \doi{10.1103/PhysRevE.94.022104}.
\newblock URL \url{http://link.aps.org/doi/10.1103/PhysRevE.94.022104}.

\bibitem[Gogolin and Eisert(2016)]{Gogolin16}
Christian Gogolin and Jens Eisert.
\newblock Equilibration, thermalisation, and the emergence of statistical
  mechanics in closed quantum systems.
\newblock \emph{Reports on Progress in Physics}, 79\penalty0 (5):\penalty0
  056001, 2016.
\newblock \doi{10.1088/0034-4885/79/5/056001}.
\newblock URL
  \url{http://iopscience.iop.org/article/10.1088/0034-4885/79/5/056001}.

\bibitem[Popescu et~al.(2006)Popescu, Short, and Winter]{Popescu06}
Sandu Popescu, Anthony~J. Short, and Andreas Winter.
\newblock Entanglement and the foundations of statistical mechanics.
\newblock \emph{Nature Physics}, 2:\penalty0 745--758, 2006.
\newblock \doi{10.1038/nphys444}.
\newblock URL \url{http://dx.doi.org/10.1038/nphys444}.

\bibitem[Brand\~ao et~al.(2013)Brand\~ao, Horodecki, Oppenheim, Renes, and
  Spekkens]{Brandao13}
Fernando G. S.~L. Brand\~ao, Micha\l{} Horodecki, Jonathan Oppenheim, Joseph~M.
  Renes, and Robert~W. Spekkens.
\newblock Resource theory of quantum states out of thermal equilibrium.
\newblock \emph{Physical Review Letters}, 111:\penalty0 250404, Dec 2013.
\newblock \doi{10.1103/PhysRevLett.111.250404}.
\newblock URL \url{http://link.aps.org/doi/10.1103/PhysRevLett.111.250404}.

\bibitem[Dahlsten et~al.(2011)Dahlsten, Renner, Rieper, and Vedral]{Dahlsten11}
Oscar C.~O. Dahlsten, Renato Renner, Elisabeth Rieper, and Vlatko Vedral.
\newblock Inadequacy of von {N}eumann entropy for characterizing extractable
  work.
\newblock \emph{New Journal of Physics}, 13\penalty0 (5):\penalty0 053015,
  2011.
\newblock \doi{10.1088/1367-2630/13/5/053015}.
\newblock URL
  \url{http://iopscience.iop.org/article/10.1088/1367-2630/13/5/053015}.

\bibitem[{\AA}berg(2013)]{Aberg13}
Johan {\AA}berg.
\newblock Truly work-like work extraction via a single-shot analysis.
\newblock \emph{Nature Communications}, 4:\penalty0 1925, 2013.
\newblock \doi{10.1038/ncomms2712}.
\newblock URL \url{http://dx.doi.org/10.1038/ncomms2712}.

\bibitem[Horodecki and Oppenheim(2013)]{Horodecki13}
Micha\l{} Horodecki and Jonathan Oppenheim.
\newblock Fundamental limitations for quantum and nanoscale thermodynamics.
\newblock \emph{Nature Communications}, 4:\penalty0 2059, 2013.
\newblock \doi{10.1038/ncomms3059}.
\newblock URL \url{http://dx.doi.org/10.1038/ncomms3059}.

\bibitem[Skrzypczyk et~al.(2014)Skrzypczyk, Short, and Popescu]{Skrzypczyk14}
Paul Skrzypczyk, Anthony~J. Short, and Sandu Popescu.
\newblock Work extraction and thermodynamics for individual quantum systems.
\newblock \emph{Nature Communications}, 5:\penalty0 4185, 2014.
\newblock \doi{10.1038/ncomms5185}.
\newblock URL \url{http://dx.doi.org/10.1038/ncomms5185}.

\bibitem[Brandao et~al.(2015)Brandao, Horodecki, Ng, Oppenheim, and
  Wehner]{Brandao15}
Fernando G. S.~L. Brandao, Micha\l{} Horodecki, Nelly Ng, Jonathan Oppenheim,
  and Stephanie Wehner.
\newblock The second laws of quantum thermodynamics.
\newblock \emph{Proceedings of the National Academy of Sciences}, 112:\penalty0
  3275--3279, 2015.
\newblock \doi{doi:10.1073/pnas.1411728112}.
\newblock URL \url{http://www.pnas.org/content/112/11/3275}.

\bibitem[\ifmmode \acute{C}\else \'{C}\fi{}wikli\ifmmode~\acute{n}\else
  \'{n}\fi{}ski et~al.(2015)\ifmmode \acute{C}\else
  \'{C}\fi{}wikli\ifmmode~\acute{n}\else \'{n}\fi{}ski,
  Studzi\ifmmode~\acute{n}\else \'{n}\fi{}ski, Horodecki, and
  Oppenheim]{Cwiklinski15}
Piotr \ifmmode \acute{C}\else \'{C}\fi{}wikli\ifmmode~\acute{n}\else
  \'{n}\fi{}ski, Micha\l{} Studzi\ifmmode~\acute{n}\else \'{n}\fi{}ski,
  Micha\l{} Horodecki, and Jonathan Oppenheim.
\newblock {Limitations on the Evolution of Quantum Coherences: Towards Fully
  Quantum Second Laws of Thermodynamics}.
\newblock \emph{Physical Review Letters}, 115:\penalty0 210403, Nov 2015.
\newblock \doi{10.1103/PhysRevLett.115.210403}.
\newblock URL \url{http://link.aps.org/doi/10.1103/PhysRevLett.115.210403}.

\bibitem[Lostaglio et~al.(2015{\natexlab{a}})Lostaglio, Korzekwa, Jennings, and
  Rudolph]{Lostaglio15}
Matteo Lostaglio, Kamil Korzekwa, David Jennings, and Terry Rudolph.
\newblock {Quantum Coherence, Time-Translation Symmetry, and Thermodynamics}.
\newblock \emph{Physical Review X}, 5:\penalty0 021001, Apr 2015{\natexlab{a}}.
\newblock \doi{10.1103/PhysRevX.5.021001}.
\newblock URL \url{http://link.aps.org/doi/10.1103/PhysRevX.5.021001}.

\bibitem[Egloff et~al.(2015)Egloff, Dahlsten, Renner, and Vedral]{Egloff15}
Dario Egloff, Oscar C.~O. Dahlsten, Renato Renner, and Vlatko Vedral.
\newblock A measure of majorization emerging from single-shot statistical
  mechanics.
\newblock \emph{New Journal of Physics}, 17\penalty0 (7):\penalty0 073001,
  2015.
\newblock \doi{10.1088/1367-2630/17/7/073001}.
\newblock URL
  \url{http://iopscience.iop.org/article/10.1088/1367-2630/17/7/073001}.

\bibitem[Lostaglio et~al.(2015{\natexlab{b}})Lostaglio, Jennings, and
  Rudolph]{Lostaglio15a}
Matteo Lostaglio, David Jennings, and Terry Rudolph.
\newblock Description of quantum coherence in thermodynamic processes requires
  constraints beyond free energy.
\newblock \emph{Nature Communications}, 6:\penalty0 6383, 2015{\natexlab{b}}.
\newblock \doi{10.1038/ncomms7383}.
\newblock URL \url{http://dx.doi.org/10.1038/ncomms7383}.

\bibitem[Bera et~al.(2017{\natexlab{b}})Bera, Ac\'{i}n, Ku\'{s}, Mitchell, and
  Lewenstein]{BeraPhilo16}
Manabendra~Nath Bera, Antonio Ac\'{i}n, Marek Ku\'{s}, Morgan Mitchell, and
  Maciej Lewenstein.
\newblock Randomness in quantum mechanics: Philosophy, physics and technology.
\newblock \emph{Reports on Progress in Physics}, 80\penalty0 (12):\penalty0
  124001, 2017{\natexlab{b}}.
\newblock \doi{10.1088/1361-6633/aa8731}.
\newblock URL \url{http://iopscience.iop.org/article/10.1088/1361-6633/aa8731}.

\bibitem[Hulpke et~al.(2006)Hulpke, Poulsen, Sanpera, Sen(De), Sen, and
  Lewenstein]{Hulpke2006}
Florian Hulpke, Uffe~V. Poulsen, Anna Sanpera, Aditi Sen(De), Ujjwal Sen, and
  Maciej Lewenstein.
\newblock Unitarity as preservation of entropy and entanglement in quantum
  systems.
\newblock \emph{Foundations of Physics}, 36\penalty0 (4):\penalty0 477--499,
  2006.
\newblock ISSN 1572-9516.
\newblock \doi{10.1007/s10701-005-9035-7}.
\newblock URL \url{http://dx.doi.org/10.1007/s10701-005-9035-7}.

\bibitem[Jarzynski(2000)]{Jarzynski2000}
Christopher Jarzynski.
\newblock Hamiltonian derivation of a detailed fluctuation theorem.
\newblock \emph{Journal of Statistical Physics}, 98\penalty0 (1):\penalty0
  77--102, Jan 2000.
\newblock ISSN 1572-9613.
\newblock \doi{10.1023/A:1018670721277}.
\newblock URL \url{https://doi.org/10.1023/A:1018670721277}.

\bibitem[Esposito et~al.(2009)Esposito, Harbola, and Mukamel]{Esposito09}
Massimiliano Esposito, Upendra Harbola, and Shaul Mukamel.
\newblock Nonequilibrium fluctuations, fluctuation theorems, and counting
  statistics in quantum systems.
\newblock \emph{Rev. Mod. Phys.}, 81:\penalty0 1665--1702, Dec 2009.
\newblock \doi{10.1103/RevModPhys.81.1665}.
\newblock URL \url{https://link.aps.org/doi/10.1103/RevModPhys.81.1665}.

\bibitem[Esposito et~al.(2010)Esposito, Lindenberg, and Van~den
  Broeck]{Esposito10}
Massimiliano Esposito, Katja Lindenberg, and Christian Van~den Broeck.
\newblock Entropy production as correlation between system and reservoir.
\newblock \emph{New Journal of Physics}, 12\penalty0 (1):\penalty0 013013,
  2010.
\newblock \doi{10.1088/1367-2630/12/1/013013}.
\newblock URL
  \url{http://iopscience.iop.org/article/10.1088/1367-2630/12/1/013013}.

\bibitem[Strasberg et~al.(2017)Strasberg, Schaller, Brandes, and
  Esposito]{Strasberg17}
Philipp Strasberg, Gernot Schaller, Tobias Brandes, and Massimiliano Esposito.
\newblock {Quantum and Information Thermodynamics: A Unifying Framework Based
  on Repeated Interactions}.
\newblock \emph{Physical Review X}, 7:\penalty0 021003, Apr 2017.
\newblock \doi{10.1103/PhysRevX.7.021003}.
\newblock URL \url{https://link.aps.org/doi/10.1103/PhysRevX.7.021003}.

\bibitem[Sagawa(2012)]{Sagawa2012}
Takahiro Sagawa.
\newblock \emph{Second law-like inequalities with quantum relative entropy: an
  introduction. {Lectures on Quantum Computing, Thermodynamics and Statistical
  Physics}}, volume~8.
\newblock World Scientific, Singapore, 2012.
\newblock \doi{10.1142/9789814425193_0003}.
\newblock URL \url{https://doi.org/10.1142/9789814425193_0003}.

\bibitem[Puglisi et~al.(2017)Puglisi, Sarracino, and Vulpiani]{Vulpiani2017}
A.~Puglisi, A.~Sarracino, and A.~Vulpiani.
\newblock Temperature in and out of equilibrium: A review of concepts, tools
  and attempts.
\newblock \emph{Physics Reports}, 709-710:\penalty0 1--60, 2017.
\newblock ISSN 0370-1573.
\newblock \doi{https://doi.org/10.1016/j.physrep.2017.09.001}.
\newblock URL
  \url{https://www.sciencedirect.com/journal/physics-reports/vol/709/suppl/C}.

\bibitem[Sparaciari et~al.(2017)Sparaciari, Oppenheim, and Fritz]{Sparaciari16}
Carlo Sparaciari, Jonathan Oppenheim, and Tobias Fritz.
\newblock Resource theory for work and heat.
\newblock \emph{Physical Review A}, 96:\penalty0 052112, Nov 2017.
\newblock \doi{10.1103/PhysRevA.96.052112}.
\newblock URL \url{https://link.aps.org/doi/10.1103/PhysRevA.96.052112}.

\bibitem[Wilming et~al.(2017)Wilming, Gallego, and Eisert]{Wilming2017}
Henrik Wilming, Rodrigo Gallego, and Jens Eisert.
\newblock Axiomatic characterization of the quantum relative entropy and free
  energy.
\newblock \emph{Entropy}, 19\penalty0 (6), 2017.
\newblock ISSN 1099-4300.
\newblock \doi{10.3390/e19060241}.
\newblock URL \url{http://www.mdpi.com/1099-4300/19/6/241}.

\bibitem[M{\"u}ller(2017)]{Mueller2017}
Markus~P. M{\"u}ller.
\newblock Correlating thermal machines and the second law at the nanoscale.
\newblock arXiv[quant-ph]:1707.03451, 2017.
\newblock URL \url{https://arxiv.org/abs/1707.03451}.

\bibitem[Jaynes(1957{\natexlab{a}})]{Jaynes57a}
Edwin~T. Jaynes.
\newblock {Information Theory and Statistical Mechanics}.
\newblock \emph{Physical Review}, 106:\penalty0 620--630, May
  1957{\natexlab{a}}.
\newblock \doi{10.1103/PhysRev.106.620}.
\newblock URL \url{https://link.aps.org/doi/10.1103/PhysRev.106.620}.

\bibitem[Jaynes(1957{\natexlab{b}})]{Jaynes57b}
Edwin~T. Jaynes.
\newblock {Information Theory and Statistical Mechanics. II}.
\newblock \emph{Physical Review}, 108:\penalty0 171--190, Oct
  1957{\natexlab{b}}.
\newblock \doi{10.1103/PhysRev.108.171}.
\newblock URL \url{https://link.aps.org/doi/10.1103/PhysRev.108.171}.

\bibitem[Pusz and Woronowicz(1978)]{Pusz78}
Wies{\l}aw Pusz and Stanis{\l}aw~L. Woronowicz.
\newblock Passive states and {KMS} states for general quantum systems.
\newblock \emph{Communications in Mathematical Physics}, 58\penalty0
  (3):\penalty0 273--290, 1978.
\newblock ISSN 1432-0916.
\newblock \doi{10.1007/BF01614224}.
\newblock URL \url{http://dx.doi.org/10.1007/BF01614224}.

\bibitem[Lenard(1978)]{Lenard78}
Andrew Lenard.
\newblock Thermodynamical proof of the {G}ibbs formula for elementary quantum
  systems.
\newblock \emph{Journal of Statistical Physics}, 19\penalty0 (6):\penalty0
  575--586, 1978.
\newblock ISSN 1572-9613.
\newblock \doi{10.1007/BF01011769}.
\newblock URL \url{http://dx.doi.org/10.1007/BF01011769}.

\bibitem[Masanes and Oppenheim(2017)]{Masanes2017}
Llu\'is Masanes and Jonathan Oppenheim.
\newblock A general derivation and quantification of the third law of
  thermodynamics.
\newblock \emph{Nature Communications}, 8:\penalty0 14538, 2017.
\newblock \doi{10.1038/ncomms14538}.
\newblock URL \url{http://dx.doi.org/10.1038/ncomms14538}.

\bibitem[{Yunger Halpern} and Renes(2016)]{Thermo-resource-charges-1}
Nicole {Yunger Halpern} and Joseph~M. Renes.
\newblock Beyond heat baths: Generalized resource theories for small-scale
  thermodynamics.
\newblock \emph{Physical Review E}, 93:\penalty0 022126, Feb 2016.
\newblock \doi{10.1103/PhysRevE.93.022126}.
\newblock URL \url{https://link.aps.org/doi/10.1103/PhysRevE.93.022126}.

\bibitem[{Yunger Halpern}(2018)]{Thermo-resource-charges-2}
Nicole {Yunger Halpern}.
\newblock {Beyond heat baths II: Framework for generalized thermodynamic
  resource theories}.
\newblock \emph{Journal of Physics A: Mathematical and Theoretical},
  51:\penalty0 094001, 2018.
\newblock \doi{10.1088/1751-8121/aaa62f}.
\newblock URL
  \url{http://iopscience.iop.org/article/10.1088/1751-8121/aaa62f/meta}.

\bibitem[Lostaglio et~al.(2017)Lostaglio, Jennings, and Rudolph]{Lostaglio2017}
Matteo Lostaglio, David Jennings, and Terry Rudolph.
\newblock Thermodynamic resource theories, non-commutativity and maximum
  entropy principles.
\newblock \emph{New Journal of Physics}, 19\penalty0 (4):\penalty0 043008,
  2017.
\newblock \doi{10.1088/1367-2630/aa617f}.
\newblock URL \url{http://iopscience.iop.org/article/10.1088/1367-2630/aa617f}.

\bibitem[Guryanova et~al.(2016)Guryanova, Popescu, Short, Silva, and
  Skrzypczyk]{Guryanova2016}
Yelena Guryanova, Sandu Popescu, Anthony~J. Short, Ralph Silva, and Paul
  Skrzypczyk.
\newblock Thermodynamics of quantum systems with multiple conserved quantities.
\newblock \emph{Nature Communications}, 7:\penalty0 12049, 2016.
\newblock \doi{10.1038/ncomms12049}.
\newblock URL \url{http://dx.doi.org/10.1038/ncomms12049}.

\bibitem[{Yunger Halpern} et~al.(2016){Yunger Halpern}, Faist, Oppenheim, and
  Winter]{Halpern2016}
Nicole {Yunger Halpern}, Philippe Faist, Jonathan Oppenheim, and Andreas
  Winter.
\newblock Microcanonical and resource-theoretic derivations of the thermal
  state of a quantum system with noncommuting charges.
\newblock \emph{Nature Communications}, 7:\penalty0 12051, 2016.
\newblock \doi{10.1038/ncomms12051}.
\newblock URL \url{http://dx.doi.org/10.1038/ncomms12051}.

\bibitem[Khanian et~al.(2018)Khanian, Nath~Bera, Riera, Lewenstein, and
  Winter]{MultipleChargesThms}
Zahra~B. Khanian, Manabendra Nath~Bera, Arnau Riera, Maciej Lewenstein, and
  Andreas Winter.
\newblock Resource theory of work and heat and everything else: basing
  thermodynamics of multiple non-commuting conserved quantities on an
  asymptotic equivalence principle.
\newblock In preparation, 2018.

\end{thebibliography}

\vfill\pagebreak

\appendix

\section{Alternative formulation of the \protect\\ Clausius statement for the second law}
\label{App:C}
\begin{lem}[Clausius statement]
No \aw{iso-entropic} equilibration process is possible whose \aw{\emph{sole}} result is the 
transfer of bound energy (i.e. heat) from an equilibrium state with inverse temperature 
$\beta_C$ to another equilibrium state with \aw{inverse temperature $\beta_H < \beta_C$}.
\end{lem}

\begin{proof}
In order to prove it, we need to show that the \aw{iso-entropic} equilibration process 
cannot lead to \aw{a negative} $\Delta \beta = \beta_H-\beta_C$. 
\aw{Further, that} any \aw{iso-entropic} process that leads to such increase must \aw{require} a supply of energy. 
First consider the case where $\beta_C = \beta_H$ ($\Delta \beta=0$), which means the corresponding states 
$\gamma_C$ and $\gamma_H$ are in equilibrium between each other and jointly \aw{in} the min-energy state. 
The min-energy principle says that any \aw{iso-entropic} transformation to make $\beta_C \neq \beta_H$ is 
bound to increase the sum of their individual min-energy. Therefore it requires \aw{an} inflow of energy, 
which is nothing but introduction of work.  

Note that for a completely passive state $\gamma$ with given Hamiltonian $H$ and $\beta$, the change in 
bound energy, which is nothing but heat in our definition, due to an infinitesimal change in entropy $d S(\gamma)$, is given by 
\begin{align}
 d E(\gamma)= \frac{1}{\beta} \ d S(\gamma).
\end{align}
Therefore \aw{the} larger $\beta$, \aw{the} smaller the change in internal energy for a fixed infinitesimal 
change in entropy \aw{of} the state. Now \aw{assume by contradiction that} $\beta_C > \beta_H$. 
Any flow of heat from $\gamma_C$ to $\gamma_H$ will lead to a reduction of entropy in the former. 
That will also lead to an equal increase of of the same in the latter. A very small amount of 
bound energy will result in a infinitesimal entropy flow, say $d S$ from $\gamma_C$ to $\gamma_H$. 
However the change in internal energy ($d E(\gamma_{C/H})=\frac{d S}{\beta_{C/H}}$) would not be 
equal and that is $-d E(\gamma_C) < d E(\gamma_H)$, for $\beta_C > \beta_H$. As a consequence, 
in this \aw{iso-entropic} process the overall energy is bound to increase, which is not possible 
without, again, an influx of energy or work. Therefore, without external work the process will never take place spontaneously.
\end{proof}

\section{Alternative formulation of the Kelvin-Planck statement for the second law}
\label{App:KP}
\begin{lem}[Kelvin-Planck statement]
No \aw{iso-entropic} equilibration process is possible whose \aw{\emph{sole}} result is the absorption 
of bound energy (heat) from an equilibrium state and its complete conversion into work.
\end{lem}

\begin{proof}
To prove the Kelvin-Planck statement, 
we consider \aw{the} following two \CP
states $\gamma_A$ and $\gamma_B$ with inverse temperatures $\beta_A$
and $\beta_B$, respectively, 
which undergo an 
\aw{iso-entropic} equilibration transformation as
\begin{align}
  \gamma_A \otimes \gamma_B \otimes \proj{0}_W \xrightarrow{\Lambda^{EP}_{ABW}} \gamma_{AB} \otimes \proj{W}_W,
\end{align}
\aw{where} $\gamma_{AB}$ is the final joint equilibrium state with \aw{inverse temperature} 
$\beta_{AB}$. Since the initial states are \CP states, a non-zero $W$ can only result from heat transfer. 
Let us say $\beta_A < \beta_B$, thus $\beta_A < \beta_{AB} < \beta_B$. In this case the heat \aw{flows out} 
from $A$, say by an amount $\Delta Q_A$, with an associated decrease in its entropy. 
Again the information conservation of whole process guarantees an entropy increase in $B$. 
Therefore, there has to be an associated increase in bound energy (heat) content in $\gamma_B$. 
As a result, a part of $\Delta Q_A$ could be converted to work and that is 
\begin{align}
  W\leqslant \Delta Q_A - \Delta Q_B.
\end{align}
Note \aw{that} any transfer of heat, i.e.~$-\Delta Q_A >0$, also bounds $\Delta Q_B >0$. 
As a consequence, heat cannot be converted into work completely. 
\end{proof}

\end{document}